\documentclass[12pt]{article}
\usepackage[T1]{fontenc}
\usepackage{mathpazo}

\usepackage{amsmath,commath,amsfonts,amsthm,amssymb,mathtools,mathrsfs}
\usepackage{enumitem}
\usepackage[margin=1in]{geometry}
\usepackage{comment}
\usepackage{float}
\usepackage{natbib}
\usepackage{booktabs}
\usepackage{siunitx}
\usepackage{dcolumn}
\newcolumntype{d}[1]{D{.}{.}{#1}}

\usepackage{setspace}
\usepackage{multirow}
\usepackage[flushleft]{threeparttable}
\usepackage{tikz}
\usepackage{pgfplots}
\usepgfplotslibrary{fillbetween}

\usepackage{caption}
\usepackage{xcolor}
\usepackage[colorlinks = true,linkcolor = blue]{hyperref}
\hypersetup{
	colorlinks = true,
	linkcolor = blue,
	citecolor = blue,
}
\usepackage{subcaption}
\DeclareMathOperator*{\argmax}{argmax}

\newtheorem{proposition}{Proposition}

\newtheorem{lemma}{Lemma}
\newtheorem{assumption}{Assumption}

\newtheorem{corollary}{Corollary}
\newtheorem{example}{Example}
\newtheorem*{remark}{Remark}

\makeatletter
\def\@xfootnote[#1]{%
  \protected@xdef\@thefnmark{#1}%
  \@footnotemark\@footnotetext}
\makeatother

\makeatletter
\def\titlefootnote{\ifx\protect\@typeset@protect\expandafter\footnote\else\expandafter\@gobble\fi}
\makeatother

\begin{document}

\title{Personalized Subsidy Rules\thanks{Alphabetical ordering; both authors contributed equally to this work. We are grateful to Graham Elliott, Yixiao Sun, and Kaspar W\"uthrich for their constant support of this paper. We also thank Le-Yu Chen, Wei-Lin Chen, Chen Lin, Ming-Jen Lin, and participants at the UC San Diego econometrics seminar and the National Taiwan University seminar.}}

\author{Yu-Chang Chen\thanks{Department of Economics, University of California, San Diego. Email: yuc391@ucsd.edu.} \quad \quad Haitian Xie\thanks{Department of Economics, University of California, San Diego. Email: hax082@ucsd.edu.} }

\date{\textbf{\today}}

\maketitle
\thispagestyle{empty}
\vspace{-2em}

\begin{abstract} 
	\onehalfspacing
    Subsidies are commonly used to encourage behaviors that can lead to short- or long-term benefits. Typical examples include subsidized job training programs and provisions of preventive health products, in which both behavioral responses and associated gains can exhibit heterogeneity. This study uses the marginal treatment effect (MTE) framework to study personalized assignments of subsidies based on individual characteristics. First, we derive the optimality condition for a welfare-maximizing subsidy rule by showing that the welfare can be represented as a function of the MTE. Next, we show that subsidies generally result in better welfare than directly mandating the encouraged behavior because subsidy rules implicitly target individuals through unobserved heterogeneity in the behavioral response. When there is positive selection, that is, when individuals with higher returns are more likely to select the encouraged behavior, the optimal subsidy rule achieves the first-best welfare, which is the optimal welfare if a policy-maker can observe individuals' private information. 
 	We then provide methods to (partially) identify the optimal subsidy rule when the MTE is identified and unidentified. Particularly, positive selection allows for the point identification of the optimal subsidy rule even when the MTE curve is not. As an empirical application, we study the optimal wage subsidy using the experimental data from the Jordan New Opportunities for Women pilot study.

	
\end{abstract}

\newpage
\section{Introduction}

Governments worldwide have offered various subsidies, such as tuition subsidies for college education, price subsidies for preventive health products, and childcare subsidies for certified daycare services, to encourage self-serving and socially beneficial behaviors among households and individuals. Owing to its relevance and prevalence, many studies across fields in economics have contributed to the evaluation of subsidy programs. For example, in labor economics, empirical works have used exogenous variations in the choice of education to estimate the subsidy on return to tuition \citep[][]{ichimura2002semiparametric, carneiro2011estimating}. In development economics, researchers have examined the impact of subsidies on the take-up of insecticide-treated bed nets \citep[][]{cohen2010free}.
These studies demonstrate the importance of informing the policy-maker on the optimal allocation of subsidies. 

In this paper, we examine subsidy rules that provide personalized subsidies to maximize the targeted population's welfare. A \emph{subsidy rule}, also called (subsidy-based) policy, is defined as an assignment rule that maps individuals to the amount of subsidy based on their observable characteristics. As the subsidy's effect on both the take-up and welfare outcomes could vary across individuals, an efficient subsidy scheme should consider the welfare effect, behavioral response, and cost of subsidies altogether. Accordingly, an ideal allocation of subsidies increases the take-up among those who benefit the most from it while avoiding those unlikely to benefit from the subsidy. As an advantage, our approach allows for flexible specifications of the cost functions, which can depend on the take-up, individual characteristics, and the amount of subsidy.



We adopt the \emph{marginal treatment effect} (MTE) framework \citep{heckman2005structural} to study the subsidy problem. In our setting, the subsidy is an instrumental variable in the following sense. First, the subsidy only affects the take-up of the subsidized behavior (the treatment) but is otherwise excluded from the outcome of interest. In case of tuition, we are effectively assuming that the subsidy affects a student's future income through education received in school. Second, the subsidy is assumed to be randomly assigned (conditional on the covariates) to the population from which our data are sampled, which holds true for any experiment conducted using subsidies. The exclusion restriction and exogeneity of the subsidy help in causally identifying the treatment effect and the individuals' treatment selection process, which in turn helps in identifying the subsidy rule that maximizes the counterfactual welfare by allocating subsidies based on observed heterogeneity.


The MTE framework is appropriate for studying the subsidy problem because of its straightforward yet flexible method for modeling the treatment effects and the treatment selection process. Explicit modeling of the subsidy as an instrument in the treatment selection process helps improve both the interpretability of the results and the practicability of our procedure. Structural assumptions backed by economic theory can be easily incorporated into the model to facilitate identification and estimation of the optimal subsidy rule. As an example, when individuals with higher returns are more likely to select the treatment is the case of positive selection. This can be modeled by assuming that the MTE is decreasing. Our study shows that there are numerous implications of a monotonic MTE in the optimal subsidy problem.

The subsidy problem is analyzed by first characterizing the welfare of subsidy rules using the MTE curve. The MTE can serve as the building block for other treatment parameters such as \emph{average treatment effect} (ATE) and \emph{policy-relevant treatment effect} (PRTE) \citep{heckman2005structural} after suitable reweighting. Consistent with these results, we show that the welfare of a given subsidy rule can also be expressed as a function of the MTE. The welfare characterization can be used to derive necessary conditions for a subsidy rule to be optimal. We show that the necessary condition becomes sufficient when there is positive selection.

Subsidy-based policies are popular because policies that directly mandate the treatment, termed \emph{direct policies}, are not always viable to the policy-maker owing to practical issues. This study shows that assigning subsidies also achieves higher welfare than directly mandating the treatment, which justifies the subsidy rules. Subsidy rules enjoy this property because, as individuals make treatment choices, subsidy-based policies implicitly target individuals based on their unobserved (to the policy-maker) heterogeneity. For example, individuals make schooling decisions based on any likely private information regarding their returns to higher education. Consequentially, different subsidies may draw students with different returns. If the selection into education is positive, smaller subsidy amounts tend to draw students with higher returns when all other things are equal. Therefore, if carefully designed, a subsidy-based policy can leverage individuals' private information on their returns and may achieve better welfare than direct policies that neglect this information. 

The welfare of subsidy rules has another surprising characteristic besides the advantage of subsidy-based policies compared with direct policies. We show that, under positive selection, the optimal subsidy rule achieves the \emph{first-best} welfare, which is defined as the highest attainable welfare if the policy-maker can observe the individuals' private information used in the treatment choice decisions. In practice, it is not feasible for a policy-maker to observe private information. However, using subsidies, the policy-maker can simulate the effect of the \emph{infeasible policies} that target unobserved information. These subsidy rules are the best because when the MTE is decreasing, the subsidy rule can replicate the (infeasible) first-best policy using individuals' self-selection into the treatment.

Next, we identify the optimal policy.
Considering the result of welfare representation, the identification problem becomes straightforward if the MTE is known. However, point identification of an entire MTE curve requires the support of a propensity score to be large, which is often impossible. We study two approaches to circumvent this issue. First, under positive selection, we can identify the optimal subsidy provided that we can identify the zero of the MTE curve. Second, we can incorporate other shape restrictions on the MTE and identify a partial ranking among the subsidy rules similar to \cite{kasy2016partial}.

From a practical perspective, this study addresses the problem of learning the optimal allocation of subsidies from (quasi-) experimental data. Our characterization result simplifies the optimal subsidy problem to the identification of the MTE curve. From this perspective, this study bridges the gap between the problems of policy evaluation and policy design. In particular, the dual role of subsidies emphasized in this study is also relevant to the application. When estimating the effects, subsidies serve as instrumental variables for estimating the treatment effects. When used as policy tools, subsidies can serve as the subject of assignments. 

We illustrate the theoretical results by an empirical application.
Using the experimental data from \cite{groh2016wage}, we apply our method to identify the optimal wage subsidies for female college students in Jordan. Although \cite{groh2016wage} concluded that wage subsidies are ineffective for increasing the long-term labor market participation, our analysis indicates that their conclusion may instead be a consequence of an inefficient allocation of subsidies in the experiment. First, we find that in the experiment, the subsidy is substantially higher than the optimal amount suggested by our method. Thus, individuals with low returns may be negatively selected, leading to a lesser effect of subsidies identified in the experiment. We also find that targeting students' majors can substantially increase efficiency. Specifically, medical school students tend to receive greater benefits from a subsidy, and the amount assigned to them should be different from students of other backgrounds. Therefore, instead of concluding that wage subsidy is ineffective, our result shows that well-designed wage subsidies can be an effective tool for boosting female long-term labor market participation.

This study is presented as follows. The remaining part of this section discusses the literature. Section \ref{sec:setup} introduces the model setup and the subsidy problem. Section \ref{sec:welfare_opt_policy} presents the welfare representation results and the optimality conditions for welfare-maximizing subsidy rules. In Section \ref{sec:welfare_property}, we compare the optimal welfare of subsidy rules with other types of policies, including direct and infeasible policies. Section \ref{sec:weflare_ranking} presents a set of results regarding the identifications of optimal subsidies. 
Section \ref{sec:empirical} presents the empirical application. Section \ref{sec:conclusion} concludes. All proofs are collected in Appendix \ref{sec:proofs}.


\subsection{Connection to the literature}

This study contributes to a growing literature on personalized treatment rules, including \cite{manski2004statistical,dehejia2005program, hirano2009asymptotics,stoye2009minimax,bhattacharya2012inferring,kitagawa2018should,athey2021policy}. For a recent review of the subject, see \cite{hirano2019statistical}. Most studies have focused on the assignment of treatments; however, the case of assigning subsidies as encouragement to treatment is understudied.\footnote{The only exception found is \cite{qiu2020optimal}, who study the optimal assignment of binary instrumental variables as ``encouragements'' to treatment take-ups. This study, however, examines the case of a continuous instrument, which is more relevant, for example, when considering monetary subsidies.} Although the existing methods for treatment assignments can be applied to subsidy assignments by adopting an \emph{intention-to-treat} approach that essentially assigns subsidies based on reduced-form estimates of subsidy effects \citep[][]{bhattacharya2012inferring,kitagawa2018should}, we argue that a selection-model approach that explicitly models treatment choice as a function of the assigned subsidy can have its advantages. First, our approach allows the realized cost of subsidies to depend on treatment take-ups, which the intention-to-treat approaches cannot as there is no model of treatment take-ups. Second, by explicitly modeling treatment choice, we can incorporate shape restrictions in the selection equation to aid identification and estimation, such as mandating positive effects of the subsidy on treatment take-ups \citep{horowitz2017nonparametric}. Third, the examined subsidy is continuous, while studies on treatment assignments have mostly focused on the case of a binary treatment.

Although most studies on individualized treatment rules consider policy learning under unconfoundedness, recent studies have examined cases when endogeneity arises for reasons such as noncompliances or omitted variable bias \citep{kasy2016partial,cui2020semiparametric,qiu2020optimal,athey2021policy, Byambadalai2021Identification, pu2021estimating}. We use instrumental variables for identifications, and our study is no exception. Similar to \cite{kasy2016partial}, we study the welfare rankings of policies when the treatment effect is only partially identified, although we focus on the assignment of instruments rather than the treatment itself. \cite{pu2021estimating} introduced IV-optimality, a new notion of optimality, for treatment assignment rules that maximize the worst-case welfare among the identification regions. They also derive a bound on the loss in the welfare of IV-optimal rules relative to the first-best rule, which assigns treatments whenever the effect is positive. This study shows that the optimal subsidy rule outperforms the first-best policy that assigns treatments.

This study analyzes the subsidy rules using the marginal treatment effects (MTE) framework \citep{heckman1999local, heckman2001local, heckman2005structural, heckman2007econometric}. The MTEs can be identified by the method of \emph{local instrumental variables} and can be used for predicting the effects of hypothetical policies. Studies have proposed new approaches to its identification and estimation \citep{carneiro2009estimating, brinch2017beyond, mogstad2018using, mogstad2020policy, sasaki2021estimation} and to apply MTE framework to various research topics, such as unconditional quantile effects \citep{martinez2020identification} and external validity \citep{kowalski2018examine}. Among these studies, our study is most closely related to \cite{sasaki2020welfare}, which also applies the MTE framework to statistical decision rules. However, they focus on the assignment of treatment instead of subsidies. They apply the MTE framework to the method of empirical welfare maximization \citep{kitagawa2018should}, where policies are assumed to lie in a known policy class restricted to avoid complexity.\footnote{Namely, the policy class has a finite $VC$-dimension.} We do not restrict our candidate policies. Additionally, we emphasize the identification and welfare properties, while \cite{sasaki2020welfare} emphasized the estimation.

\section{The Model} \label{sec:setup}
In this section, we introduce the model's setup, including the optimal subsidy problem faced by a policy-maker. Importantly, the subsidy rule only affects the social welfare through its effect on the behavior response, which can eventually impact the outcome. We begin by introducing the MTE framework.

\subsection{Data-generating process: the MTE framework}

There are two treatment statuses, $1$ and $0$, referred to as \emph{treated} and \emph{untreated}, respectively.
Let $Y_1$ and $Y_0$ be the potential outcomes under the treatment status $1$ and $0$, respectively. The potential outcomes are related to the observable covariates as
\begin{align}\label{eqn:outcome_equation} 
Y_1 = \mu_1(X,U_1) \text{, and } Y_0 = \mu_0(X,U_0),
\end{align}
where $X$ is a vector of the observed covariates that affect the potential outcomes, $\mu_1$ and $\mu_0$ are unknown functions, and $U_1$ and $U_0$ are unobserved random variables. Let $D$ denote the binary variable that indicates the treatment status. Specifically, $D = 1$ if an individual receives the treatment and $D=0$ if an individual does not receive the treatment. The realized outcome is $Y = DY_1 + (1-D)Y_0$. 

In our setup, we distinguish two types of instruments. The first type of instrument, denoted by $Z$, is an instrument (or subsidy) randomly assigned in the data but could be manipulated by the policy-maker as policy tools. Specifically, the variable $Z$ has two roles. First, $Z$ is an instrumental variable exogenously set in the data and can facilitate the identification of treatment effects. Second, $Z$ is the subsidy, a policy tool that the policy-maker can use to influence individuals' treatment take-ups. 

The second type of instrument, denoted by $W$, is an instrument that only aids in the identification of treatment effects and not subject to the policy-maker's control. Generally, the existence of $W$ can enlarge the identification region of the MTE curve and help identify the optimal subsidy rule. However, to apply our method, one must have a nonmanipulatable instrument $W$. 

The treatment take-up is modeled by a latent-index utility model, where the selection into the treatment status depends on the individual characteristics and the instrumental variables. Given $(X,W,Z)$, the treatment take-up $D$ is determined by
\begin{align} \label{eqn:treatment_data} 	
D = \mathbf{1}\{g(X,W,Z)  \geq U_D\}, 
\end{align}
where $U_D$ is the unobserved heterogeneity in the treatment selection process. We can interpret $U_D$ as resistance to treatment take-ups: holding $(X,W,Z)$ fixed, individuals with lower $U_D$ are more likely to select into the treatment. We allow $U_D$ to correlate with $(U_1,U_0)$. That is, the resistance $U_D$ represents an individual's private information regarding the potential outcomes $(Y_1,Y_0)$. The individual uses the private information to aid the treatment choice decision as modeled by Equation (\ref{eqn:treatment_data}). In practice, the policy-maker does not observe the private information $U_D$.

The following example illustrates the different variables introduced earlier.

\begin{example}[Tuition Subsidy]
	In terms of the tuition subsidy, $Y$ can be considered earnings after graduation, $X$ as individual characteristics such as family background, $D$ as levels of education, $W$ as proximity to colleges, and $Z$ as the tuition for attending public college \citep{kane1995labor}. While the government has no direct control over students' place of residence, policy-makers may change the tuition subsidy to encourage college enrollment.
\end{example}


The most common assumptions followed in the MTE literature are used in this study.

\begin{assumption}[Random Assignment]\label{assu:random_assignment}
Conditional on $X$, the instrumental variables $(W,Z)$ are independent of the unobserved variables $(U_1,U_0,U_D)$.
\end{assumption}

\begin{assumption}[Rank Condition] \label{assu:nontrivial}
	The propensity score $g(x,w,z)$ is a nontrivial function of $(w,z)$, given any $x$.
\end{assumption}

As shown by \cite{vytlacil2002independence}, the MTE model of subsidy characterized by Equations (\ref{eqn:outcome_equation}) and (\ref{eqn:treatment_data}), combined with Assumptions \ref{assu:random_assignment} and \ref{assu:nontrivial}, is equivalent to the \cite{imbens1994identification} assumptions of independence and monotonicity for the local average treatment effect (LATE) interpretation of IV estimands.\footnote{The equivalence result in \cite{vytlacil2002independence} is derived based on $X = x$. When the monotonicity is global across all values of $x$, \cite{chen2021global} showed that $g(x,w,z)$ must be additively separable between $x$ and $(w,z)$.}

\begin{assumption} [Moment Existence]\label{assu:moment_existence}
	The expectations $E[Y_1]$ and $E[Y_0]$ exist, that is, $E[Y_1] < \infty$ and $E[Y_0]<\infty$.
\end{assumption}
	
\begin{assumption}[Density Existence]\label{assu:density_existence}
	The distribution of $U_D$ is absolute continuous with respect to the Lebesgue measure for every $X=x$.
\end{assumption}

Following Assumption \ref{assu:density_existence}, without loss of generality, we can impose the normalization that $U_D  \mid X \sim \text{Unif}[0,1]$, as $g(X,W,Z)  \geq U_D$ is equivalent to $F_{U_D \mid X} (g(X,W,Z)) \geq F_{U_D \mid X}(U_D)$, where $F_{U_D \mid X}(u) \equiv \mathbb{P}(U_D \leq u \mid X)$. 


The \emph{marginal treatment effect} (MTE) is defined as 
\begin{align} \label{eqn:def-MTE}
	\text{MTE}(x,u) \equiv \mathbb{E} \left[ Y_1 - Y_0 \mid X = x, U_D = u \right].
\end{align}
MTE has  two common interpretations: one as an average treatment effect for individuals at different margins and the other as the infinitesimal \emph{local average treatment effect} (LATE) as it is identified by the local instrumental variable \citep{heckman2005structural}. MTE corresponds to the change in population outcome resulting from an infinitesimal change in the instrumental variable. The MTE curve can be used as a building block for other conventional treatment effect parameters, such as the average treatment effect. In our analysis, MTE also plays a fundamental role. Subsequently, we can also interpret MTE as the marginal effect of increasing the subsidy. In the first set of our results, we use MTE to characterize social welfare.

By definition, $\text{MTE}(x,u)$ is the mean treatment effect for individuals with $X = x$ at the selection margin $U_D = u$, where a higher $U_D$ implies a lower willingness to select the treatment. Fixing $X = x$, a decreasing MTE curve (along the $u$-dimension) corresponds to the case of positive selection, implying that individuals who benefit more from the treatment are more likely to select it. Positive (or negative) selection can often be motivated by economic theory, econometric specifications, and empirical findings. Typical examples include the choice of education in which individuals with higher returns are more likely to invest in education. While we do not impose the assumption of monotone selection in our main results, we analyze its implications on the characterization, identification, welfare properties, and estimation of optimal policies in the relevant sections throughout this study. 


\subsection{Subsidy rules and counterfactual outcomes}

We now describe the policy problem and define the subsidy rules. In our setup, the policy-maker can influence an individual's treatment choice and thereby the realized outcome by manipulating the subsidy $Z$. Formally, a policy $\pi$ is a measurable function 
\begin{align}\label{eqn:policy_defnition}
\pi: \text{Supp}(X,W) \rightarrow \mathcal{Z}^p
\end{align}
that maps individual characteristics to the action space $Z^p$.\footnote{Although the instrumental variable $W$ does not affect the potential outcomes, the assignment of subsidy can depend on $W$ as $W$ can affect the treatment take-up. Therefore, the optimal subsidy for type $X = x$ may depend on the value of $W$ as well. It is described further in Section \ref{sec:welfare_property}. Furthermore, we restrict our attention to deterministic policies.} The action space $\mathcal{Z}^p$ is the user-specified set of the current subsidy assignments. In some of our results, we assume that the action space is equal to $\mathcal{Z}$, the support of $Z$ in the data; while in others, we allow the action space to be larger than $\mathcal{Z}$. An example of $\mathcal{Z}^p$ would be an interval $\mathcal{Z}^p = [z_l,z_u] \subset \mathbb{R}$, where the range $[z_l,z_u]$ is specified by the policy-maker. We also allow the subsidy to be negative, which represents a tax imposed by the policy-maker. We denote the set of candidate policies as $\Pi$.

Rather than directly setting a mandatory treatment assignment for each individual, a subsidy rule aims at improving the welfare by encouraging individuals to select the treatment with subsidies. Given a policy $\pi$, we assume that the counterfactual treatment choice is
\begin{align} \label{eqn:treatment_policy}
	D^\pi = \mathbf{1}\{g(X,W,\pi(X,W)) \geq U_D\}
\end{align}
and the counterfactual outcome is
\begin{align} \label{eqn:outcome}
	Y^\pi = D^\pi Y_1 + (1-D^\pi)Y_0.
\end{align}
A comparison with the case with no policy intervention in Equations (\ref{eqn:outcome_equation}) and (\ref{eqn:treatment_data}) shows the difference that the variable $Z$ is replaced by the subsidy $\pi(X,W)$. Our definition of counterfactual outcomes implicitly assume the following form of \emph{policy-invariance}: (1) the structural functions $\mu_0,\mu_1,$, and $g$, and (2) the distribution of the economic fundamentals $(X,W, U_1,U_0,U_D)$ would remain the same under the policy intervention $\pi$.

We model the cost of subsidy as follows. Let $c(x,w,z,d)$ be the cost of assigning subsidy $\pi(x,w) = z$ to $(X = x,W = w)$ individuals, who then choose treatment status $D = d$. The counterfactual cost under policy $\pi$ is  
\begin{align*}
C^\pi = c(X,W,Z,D^\pi).
\end{align*}
Unlike studies that have adopted an intention-to-treat approach \citep[e.g.,][]{kitagawa2018should}, our framework allows the realized cost of subsidies to depend on individual's treatment choice $D$. Specifically, the cost can be endogenous to the individual's decision because, unlike the intention-to-treat approaches, the treatment choice $D$ is explicitly modeled. We assume that the cost function $c$ is known to the policy-maker although the realized cost for each individual is ex-ante unknown. 
The following examples present specific forms of the cost function.

\begin{example}[Constant Cost]\label{ex:constant_cost}
\cite{kitagawa2018should} studied the optimal eligibility rule for receiving subsidized training in the Job Training Partnership Act program. In their welfare calculation, they imputed the cost of the program as $\$774$ for each eligible individual, regardless of the actual take-up. In the notation of this study, their cost function is effectively $c(x,w,z,d) = c(z)$, which is only a function of the policy assignment. However, as reported in \cite{bloom1997benefits}, the program take-up varies substantially across different gender and age groups, implying that the realized cost is, in fact, heterogeneous.

\end{example} 

\begin{example}[Voucher Cost]\label{ex:voucher_cost}
Following Example \ref{ex:constant_cost}, a more realistic cost function would be 
\begin{align*}
	c(x,w,z,d) = z \cdot d,
\end{align*}
where $z$ is the amount of subsidy paid by the government. Similar to vouchers, which incur costs only when redeemed, the dependence on the treatment choice $d$ reflects that the subsidy is paid only when the individual attends the program. The heterogeneity of treatment take-up is already embedded in this cost function as both the treatment choice and the propensity score depend on the covariates, as specified in Equation (\ref{eqn:treatment_policy}).
\end{example}

Following studies on the treatment assignment \citep[][]{manski2009identification,kitagawa2018should,athey2021policy}, we adopt the additive welfare criterion to evaluate the performance of policies. Given the cost function $c$, the welfare $S(\pi)$ of a policy $\pi$ is defined as
\begin{align}\label{eqn:welfare_definition}
S(\pi) \equiv \mathbb{E}[Y^\pi] - \mathbb{E}[C^\pi].
\end{align}
The additive welfare criterion is flexible such that it can cover various social preferences by suitably transforming the outcome variable. For example, let $\nu$ be a concave function, we can accommodate inequality-averse preference by replacing $Y$ by $\nu(Y)$ \citep{atkinson1970measurement}.\footnote{By maximizing $\mathbb{E}[Y^\pi]$, the optimal policy necessarily maximizes $V(\mathbb{E}[Y^\pi])$ for any $V$ that is an increasing transformation. The invariance property helps in dealing with cases when welfare cannot be represented by simple averages of individual outcomes. For example, in the insecticide-treated bednet example, where the policy goal is to increase the usage of nets ($Y$), we can incorporate externality by choosing $V$ that maps the coverage rate $\mathbb{E}[Y^\pi]$ to the counterfactual infection rate $V(\mathbb{E}[Y^\pi])$. Such $V$ arguably exists if the externality is approximately determined by the average coverage of nets.} Alternatively, we can interpret the welfare function $\mathbb{E}[\nu(Y^\pi)] - \mathbb{E}[C^\pi] = \mathbb{E}[\nu(Y^\pi) - C^\pi]$ as the average social function in which individuals have quasi-linear preferences. 



Given the set of possible subsidy assignments $\mathcal{Z}^p$, the policy class $\Pi$, and the social welfare function $S(\pi)$, a subsidy rule $\pi^*\in\Pi$ is said to be optimal if it attains the social optimum, namely,
\[
S(\pi^*) \equiv \sup_{\pi\in\Pi} S(\pi).
\]
For the ease of expositions, we present our results with cost function $c(x,w,z,d) = z \cdot d$ in the main text. The results for general cost functions can be found in Appendix \ref{sec:proofs} and the proofs therein.

\begin{remark} The optimal policy may not be unique. For example, in a simple case where the treatment effect is always zero $(Y_1 = Y_0)$ with zero subsidy cost, any subsidy rule is optimal. Moreover, if any targeting variable is a continuous random one, new optimal rules can be developed by modifying the existing rules on a measure-zero set without changing the implied welfare.\footnote{Therefore, any characterization of optimal policies only apply outside a measure-zero set if at least one of the targeting variable is a continuous random one. }
\end{remark}

\section{Optimal Subsidy Rules} \label{sec:welfare_opt_policy}

This section presents the welfare characterization of subsidy rules using the MTE curve. We use the characterization results to derive optimality conditions of subsidy rules under different scenarios. Simple graphical illustrations are subsequently provided for welfare characterization and optimality conditions.

\subsection{Necessary conditions for optimality}
Our first result shows that the welfare of a policy can be expressed by MTE and the propensity score.

\begin{proposition} [Characterization of Welfare] \label{prop:welfare_representation}
	Under Assumptions \ref{assu:random_assignment} - \ref{assu:density_existence}, we have
	\begin{align} \label{eqn:welfare-characterization}
		\mathbb{E}[Y^\pi] = \mathbb{E}[Y_0] +  \mathbb{E} \left[ \int_0^{g(X,W,\pi(X,W))} \text{MTE}(X,u) du \right], 
		\end{align}
and
	\begin{align*}
		  \mathbb{E}[C^\pi] &= \mathbb{E}[\pi(X,W)g(X,W,\pi(X,W))] 
	\end{align*}
for the cost function $c(x,w,z,d) = z \cdot d$.\footnote{Results for the  general cost functions can be found in Appendix \ref{sec:proofs}.}
\end{proposition}

The upper bound of the integral of MTE in Equation (\ref{eqn:welfare-characterization}) clarifies the study results. This upper bound is the propensity score induced by the subsidy rule. Mathematically, the subsidy rule affects the welfare by manipulating the integration region of MTE. The optimality conditions in this section are derived by finding the upper bound that maximizes the net area between the MTE curve and the cost function. The welfare properties studied in Section \ref{sec:welfare_property} are also closely related to this upper bound.

 As shown in Proposition \ref{prop:welfare_representation}, the welfare of a subsidy rule has three parts: the baseline outcome for the untreated, the treatment effects on individuals who are induced into treatment status, and the expected costs of subsidies for the treatment takers. The baseline outcome $\mathbb{E}[Y_0]$ is irrelevant for welfare comparison between policies because it is unaffected by the policy.

Here, we examine the methods for finding the optimal policy.
The method is to maximize the welfare function point-wise for each combination of $(x,w)$. As implied by the choice Equation (\ref{eqn:treatment_policy}), if the policy-maker assigns subsidy $\pi(x,w)$ to $(X = x, W = w)$ individuals, the counterfactual take-up rate $u^\pi_{x,w}$ among these individuals would be given by 
\begin{align*}
u^\pi_{x,w} & \equiv P(D^\pi = 1 \mid X = x,W = w) \\
& = P(g(x,w,\pi(x,w) \geq U_D \mid X = x,W = w) \\
& = g(x,w,\pi(x,w)).
\end{align*}
When the propensity score $g(x,w,z)$ is monotonic in $z$, that is, an increase in subsidy always induces more individuals into the treatment, there is a one-to-one mapping between the amount of subsidy $\pi(x,w)$ and the counterfactual take-up rate $u^\pi(x,w)$. Specifically, the relationship is governed by the propensity score, and the problem can be simplified by changing the variables. We first solve the optimal take-up rate problem
\begin{align} \label{eqn:optim_prob}
	u^{\pi^*}_{x,w} \in \argmax_{u_{x,w} \in I_{x,w}} \text{ } \int_0^{u_{x,w}} \text{MTE}(x,u') du' - c(x,w,g_{x,w}^{-1}(u_{x,w}),1)\cdot u_{x,w}, 
\end{align}
where $I_{x,w} = \{ g(x,w,z): z \in \mathcal{Z}^p \}$ is the image of $g(x,w,\cdot)$, and $g_{x,w}(z) \equiv g(x,w,z)$. The set $I_{x,w}$ reflects to degree of the policy-maker's influence on treatment take-ups by manipulating the subsidy $Z$, and, $g_{x,w}^{-1}(u)$ is the amount of subsidy needed to induce a take-up rate of $u$. The integration in Equation (\ref{eqn:optim_prob}) starts from $0$ as individuals with low $U_D$ are always induced first.

To avoid technical issues, we assume $I_{x,w}$ is a closed set. We also assume that the propensity score $g_{x,w}(z)$ is strictly increasing in $z$ so that $g_{x,w}(z)$ is invertible. The monotonicity of the propensity score can be assumed in the context of subsidies. An increase in the amount of subsidies could always lead to more people participating in the treatment program. The following assumptions ensure that the optimization defined in Equation (\ref{eqn:optim_prob}) is well defined and has a solution.

\begin{assumption}[Continuity]\label{assu:continuity}
The propensity score $g(x,w,z)$ and the cost function $C(x,w,z)$ are continuous in the subsidy $z$,  and $\text{MTE}(x,u)$ is continuous in $u$. 
\end{assumption}

\begin{assumption}[Invertibility]\label{assu:invertibility}
The propensity score $g(x,w,z)$ is strictly increasing in $z$ for all $(x,w) \in \text{Supp}(X,W)$.
\end{assumption}

Although both variables $X$ and $W$ are used for targeting, they play different roles in the policy problem because $W$ is excluded from the outcome equation. Unlike $X$, which underlies the heterogeneity in the treatment effect, $W$ only enters the optimization problem through the feasible region $I_{x,w}$ and the propensity score $g(x,w,z)$. In particular, the only value of $W$ as a targeting variable comes from its effect on treatment take-up, which cannot be neglected when the policy provides incentives. In contrast, if the policy were to assign the treatment $D$ instead, targeting  variable $W$ is unnecessary as the treatment effect does not depend on $W$, as subsequently discussed.
 
After finding the optimal take-up rate $u^*_{x,w}$, we can compute the optimal subsidy level $\pi^*(x,w)$ that achieves $u^*_{x,w}$ such that
\begin{align} \label{eqn:second_step}
	g(x,w,\pi^*(x,w)) = u^*_{x,w}.
\end{align}
By construction, $u^*_{x,w}$ is always achieved by some subsidy level $\pi^*(x,w)\in\mathcal{Z}$ as $u^*_{x,w} \in I_{x,w}$. In fact, $\pi^*(x,w)$ is unique as the propensity score is strictly increasing in the amount of subsidy.

The next proposition summarizes the aforementioned arguments and states the necessary condition for a subsidy rule to be optimal.
\begin{proposition} [Optimality Condition] \label{prop:optimality_condition} Suppose that Assumptions \ref{assu:random_assignment} - \ref{assu:invertibility} hold, the cost function $c(x,w,z,d) = z \cdot d$, and the action space $\mathcal{Z}^p$ is a closed interval $[z_l,z_u]\subset\mathbb{R}$. If  $\pi^*$ is an optimal policy, then  $\pi^*$ either satisfies $\Lambda(x,w,\pi^*(x,w)) = 0$, where $\Lambda(x,w,z)$ is defined by
\begin{align}\label{eqn:opt_w_voucher}
	\Lambda(x,w,z) \equiv \text{MTE}(x,g(x,w,z)) -  z - g(x,w,z) \cdot \left[  \frac{\partial}{\partial z}g(x,w,z) \right]^{-1},
\end{align}
or satisfies $\pi^*(x,w) \in \{z_l, z_u\}$.
\end{proposition}
The result for the general cost function $c$ is included in the proof of Proposition \ref{prop:optimality_condition} in Appendix \ref{sec:proofs}. The term $\Lambda(x,w,z)$ is the marginal benefit of subsidy, which is equal to the marginal revenue $\text{MTE}(x,g(x,w,z))$ minus the marginal cost $z + g(x,w,z) \cdot \left[  \frac{\partial}{\partial z}g(x,w,z) \right]^{-1}$, which arises naturally from a monopolist's profit maximization problem. The term $g(x,w,z) \cdot \left[  \frac{\partial}{\partial z}g(x,w,z) \right]^{-1}$ is the elasticity of treatment take-up. 

Equation (\ref{eqn:opt_w_voucher}) shows that MTE can be interpreted as the average marginal benefit of increasing the amount of subsidy for individuals with $(X = x, W=w)$. Therefore, MTE is not only a treatment effect parameter as commonly understood, but also a policy-relevant parameter in the context of personalized subsidy rule.\footnote{Studies have shown a similar connection between MTE and policy effects. For example, \cite{carneiro2010evaluating} showed that the marginal policy-relevant treatment effect is a weighted average of MTE. However, in our case, MTE is shown as the marginal effects of subsidies. The distinction appears as we study personalized subsidy rules, whereas studies have considered universal changes in the amount of subsidy.}

Following is a simple example to demonstrate the optimality condition presented in Proposition \ref{prop:optimality_condition}.

 \begin{example}\label{eg:toy_example}
	Suppose $X$ is a constant and hence can be omitted. Let the cost $c$ be zero. Let $W$ and $Z$ be supported on $[0,1]$. The treatment response is $g(w,z) = \frac{1}{4}(1+z+w)$. Let the MTE be $\text{MTE}(u) = 4-2u$, which is decreasing. The image of $g(w,\cdot)$ is $I_w = [\frac{1}{4}(1+w),\frac{1}{4}(2+w)]$. Here, the optimal policy is $\pi^*(w) = 1-\frac{3}{5}w$.
\end{example}

Example \ref{eg:toy_example} shows that although the instrument $W$ is excluded from the outcome equation, $W$ is still valuable for targeting because it affects the selection into treatment.\footnote{See Section \ref{sec:welfare_property} for more discussions on whether to target $W$ for different types of policies.} In the example, the optimal subsidy is decreasing in $w$ for two reasons. First, individuals with high $w$ are ex-ante more likely to select the treatment, therefore requiring less subsidy. Second, as there is positive selection into the treatment ($\text{MTE}$ is decreasing),  inducing high-resistance (high $u$) individuals into the treatment status is less consequential. 


\subsection{Sufficient conditions for optimality when MTE is monotone}
By definition, $\text{MTE}(x,u)$ is the mean treatment effect for individuals with $X = x$ at the selection margin $U_D = u$, where a higher $U_D$ implies a lower willingness to select the treatment. Fixing $X = x$, a decreasing MTE curve (along the $u$-dimension) corresponds to the case of positive selection, implying that individuals who benefit are more likely to take the treatment. 

\begin{assumption}[Positive Selection]\label{assu:positive_selection}
The selection process is said to be positive if MTE$(x,u)$ is weakly decreasing  in $u$.
\end{assumption}

\begin{assumption}[Negative Selection]\label{assu:negative_selection}
The selection process is said to be negative if MTE$(x,u)$ is weakly increasing  in $u$.
\end{assumption}

Empirical evidence supports the monotonicity of MTE. For example, in the context of return to schooling, \cite{carneiro2009estimating} used the local polynomial regression to obtain a nonparametric estimate of the MTE curve. In their study, figure 3 showed a clear downward-slopping MTE curve. Other empirical evidence includes \cite{carneiro2011estimating, cornelissen2018benefits}. The monotonicity can also be motivated using both economic theory and econometric specifications. Following are a few such examples.

\begin{example}[Normal Selection Model] Suppose $Y_1 = X'\beta_1+ U_1$, $Y_0 = X'\beta_0+ U_0$, and $D=\mathbf{1}\{Z'\theta\geq U_D\}$. Further, assume that $(U_1, U_0, U_D)$ is jointly normally distributed and independent of $(X,Z)$, and the variance of $U_D$ is normalized to one. Then $\text{MTE}(x,u) = x'(\beta_1-\beta_0) + (\sigma_{1D}-\sigma_{0D})\Phi^{-1}(u)$, where $\sigma_{1D} = Cov(U_1, U_D), \sigma_{0D}=Cov(U_0,U_D)$, and where $\Phi^{-1}$ is the inverse of the standard normal cumulative function. For extensions to non-normal selection models, see \cite{heckman2003simple}. 
\end{example}

\begin{example}[Roy Model] In the \cite{roy1951some} model, the treatment take-up is fully determined by the potential gain, given that $D = \mathbf{1}\{\Delta \geq 0\}$, where $\Delta \equiv Y_1 - Y_0 $. Let $U_D = F_{\Delta}(\Delta)\sim\text{Unif}[0,1]$ be the normalized gain.  Then $\text{MTE}(u)=\mathbb{E}[Y_1 - Y_0|U_D=u] = F_\Delta^{-1}(u)$. 
\end{example}

\begin{example}[Generalized Roy Model with Positive Selection] \label{eg:general-roy}
	 Consider a selection model in which the treatment take-up is partially determined by the potential gain in the form $D = \mathbf{1}\{\phi(X,W,Z,\Delta,V) \geq 0\}$, where $\Delta $ is the individual treatment effect defined in the previous example, and V represents the unobserved heterogeneity. In Appendix \ref{sec:primitive_monotone}, we show that if the function $\phi(X,W,Z,\Delta,V)$ is increasing in $\Delta$, then we can construct a function $g$ and a random variable $U_D\sim\text{Unif[0,1]}$ such that (1) $U_D \perp (W,Z) \mid X$, (2) $D = \mathbf{1}\{g(X,W,Z) \geq U_D\}$, and (3) $\text{MTE} (x,u) = \mathbb{E}[\Delta|X = x, U_D = u]$ is decreasing in $u$.
\end{example}

The following proposition characterizes the optimal policy when the selection is monotone.

\begin{proposition}[Optimality under Positive Selection] \label{prop:positive_seleciton} Suppose that Assumptions \ref{assu:random_assignment} - \ref{assu:positive_selection} hold, the action space $\mathcal{Z}^p$ = $[z_l,z_u]$, and $g(x,w,z)$ is weakly concave in $z$. Further assume that the cost function $c(x,w,z,d) = z \cdot d$. Then the optimal subsidy $\pi^*(x,w)$ is given by
\begin{equation} \label{eqn:positive_selection}
    \pi^*(x,w) = 
    \begin{cases*}
		z_{l}, & \text{if $\Lambda(x,w,z)  < 0$ for all $z \in [z_l,z_u]$}, \\
      z^*, & if $\Lambda(x,w,z^*) = 0$ \text{ for some } $z^*\in [z_l,z_u]$, \\
      z_u, & if $\Lambda(x,w,z)  > 0 \text{ for all } z \in [z_l,z_u]$, \\
    \end{cases*}
  \end{equation}
  where $\Lambda(x,w,z)$ is defined in (\ref{eqn:opt_w_voucher}).  
\end{proposition}

Proposition \ref{prop:positive_seleciton} is a complete characterization of the optimal policy as one of the three cases in Equation (\ref{eqn:positive_selection}) must hold. When the selection is positive, the marginal return of subsidy decreases because individuals with higher returns are always induced first. Corner solutions arise if the marginal return is always positive or negative. We can also characterize the optimal policy for the case of negative selection, although under the assumption of no cost $c(x,w,z,d) = 0$.

\begin{proposition} [Optimality under Negative Selection]\label{prop:negative_seleciton} Suppose that Assumptions \ref{assu:moment_existence} - \ref{assu:invertibility} and \ref{assu:negative_selection} hold, the action space $\mathcal{Z}^p$ = $[z_l,z_u]$, and $\text{MTE}(x,u)$ is weakly increasing in $u$. Furthermore, assume $c(x,w,z,d) = 0$. Then, the optimal subsidy $\pi^*(x,w)$ is given by
\begin{equation}
    \pi^*(x,w) =
    \begin{cases}
	z_l, & \text{ if } \displaystyle \int^{g(x,w,z_u)}_{g(x,w,z_l)} \text{MTE}(x,u) du \leq 0, \\
	z_u, & \text{ otherwise.}
    \end{cases}
  \end{equation}
\end{proposition}

In the case of negative selection, individuals who least benefit from the treatment are always induced first, and the marginal return of subsidy is increasing. Therefore, unless the treatment effect is zero, the optimal subsidy is always a corner solution-it either assigns the highest subsidy under consideration so that individuals with high returns are persuaded to take up the treatment; or it assigns the least amount of subsidy to minimize the potential harm caused by the treatment for low-return individuals.

\subsection{Graphical Illustration} \label{subsec:graph}

The results presented in this section are illustrated using graphs. For simplicity, we make the following assumptions: the cost $c = 0$, the potential outcome $Y_0 = 0$, and $X$ and $W$ are constants and hence can be omitted in the discussion. Consequently, the subsidy rule $\pi$ becomes a scalar constant. Under these assumptions, the welfare characterization in Proposition \ref{prop:welfare_representation} can be simplified to an integral of MTE from $0$ to $g(\pi)$:
\begin{align*}
	S(\pi) = \mathbb{E}[Y^\pi] - \mathbb{E}[C^\pi] = \int_0^{g(\pi)} \text{MTE}(u) du.
\end{align*}
Figure \ref{fig:welfare-characterization} demonstrates the welfare characterization in Proposition \ref{prop:welfare_representation} and the optimality condition in Proposition \ref{prop:optimality_condition}. Figure \ref{fig:monotone-MTE} demonstrates the optimal subsidy rule under positive selection.

\begin{figure}[!htbp]
	\begin{center}
		\begin{subfigure}[]{0.5\textwidth}
			\begin{tikzpicture}
	
				\begin{axis}[
					legend style={nodes={scale=0.8, transform shape}},
					axis y line=middle, 
					axis x line=middle,
					y axis line style={opacity=0},
					ytick={0},
					xtick={0,1},
				]
				\addplot [
					thick,
					name path=mte,
					domain=0:1, 
					samples=100, 
					color=black,
				]
				{0.399691 +   4.956392 * x  -74.935120*x^2 + 275.423206*x^3 -426.853331*x^4+  298.923460*x^5  -78.211517*x^6};
				\path[name path=axis] (axis cs:0,0) -- (axis cs:1,0);
				\addplot [
					thick,
					color=blue,
					fill=blue, 
					fill opacity=0.5
				]
				fill between[
					of=mte and axis,
					soft clip={domain=0:0.2},
				];
				\addplot [
					thick,
					color=red,
					fill=red, 
					fill opacity=0.5
				]
				fill between[
					of=mte and axis,
					soft clip={domain=0.2:0.4},
				];
			
				\addplot[black, mark=diamond*] coordinates{(0.4,0)} node[above, color=black] {$g(\pi_1)$};
			
				after end axis/.code={
					\path (axis cs:0,0) node [anchor=north west,yshift=-0.075cm,xshift=-0.075cm] {0};
					
				}
				
				\end{axis}
				
				\end{tikzpicture}
				\caption{Welfare of subsidy $\pi_1$}
		\end{subfigure}%
		\begin{subfigure}[]{0.5\textwidth}
			\begin{tikzpicture}
	
				\begin{axis}[
					legend style={nodes={scale=0.8, transform shape}},
					axis y line=middle, 
					axis x line=middle,
					y axis line style={opacity=0},
					ytick={0},
					xtick={0,1},
				]
				\addplot [
					thick,
					name path=mte,
					domain=0:1, 
					samples=100, 
					color=black,
				]
				{0.399691 +   4.956392 * x  -74.935120*x^2 + 275.423206*x^3 -426.853331*x^4+  298.923460*x^5  -78.211517*x^6};
				\addlegendentry{MTE Curve};
				\path[name path=axis] (axis cs:0,0) -- (axis cs:1,0);
				\addplot [
					thick,
					color=blue,
					fill=blue, 
					fill opacity=0.5
				]
				fill between[
					of=mte and axis,
					soft clip={domain=0:0.2},
				];
				\addlegendentry{Positive Welfare};
				\addplot [
					thick,
					color=red,
					fill=red, 
					fill opacity=0.5
				]
				fill between[
					of=mte and axis,
					soft clip={domain=0.2:0.5},
				];
				\addlegendentry{Negative Welfare};
				\addplot [
					thick,
					color=blue,
					fill=blue, 
					fill opacity=0.5
				]
				fill between[
					of=mte and axis,
					soft clip={domain=0.5:0.7},
				];
				\addplot[black, mark=diamond*] coordinates{(0.7,0)} node[below, color=black] {$g(\pi_2)$};
			
				after end axis/.code={
					\path (axis cs:0,0) node [anchor=north west,yshift=-0.075cm,xshift=-0.075cm] {0};
					
				}
				
				\end{axis}
				
				\end{tikzpicture}
				\caption{Welfare of subsidy $\pi_2$}
		\end{subfigure}

	\end{center}
	\caption{Welfare Characterization}
	\label{fig:welfare-characterization}
	\caption*{\footnotesize The two graphs demonstrate the welfare under two arbitrary subsidy rules, $\pi_1$ and $\pi_2$. In each case, the welfare is an integral of MTE from $0$ to $g(\pi)$, which is equal to the area of the blue region minus that of the red region. To minimize the area of the red region, we need $\text{MTE}(g(\pi^*)) = 0$, which is the necessary condition in Proposition \ref{prop:optimality_condition}.}
\end{figure}

\begin{figure}[!htbp]
	\begin{center}
		\begin{subfigure}[]{.5\textwidth}
			\begin{tikzpicture}
	
				\begin{axis}[
					legend style={nodes={scale=0.8, transform shape}},
					axis y line=middle, 
					axis x line=middle,
					y axis line style={opacity=0},
					ytick={0},
					xtick={0,1},
				]
				\addplot [
					thick,
					name path=mte,
					domain=0:1, 
					samples=100, 
					color=black,
				]
				{cos(deg(x)*3.1415)+0.25};
				\path[name path=axis] (axis cs:0,0) -- (axis cs:1,0);
				\addplot [
					thick,
					color=blue,
					fill=blue, 
					fill opacity=0.5
				]
				fill between[
					of=mte and axis,
					soft clip={domain=0:0.58},
				];
		
				\addplot[black, mark=diamond*] coordinates{(0.58,0)} node[above right, color=black] {$g(\pi^*)$};
			
				after end axis/.code={
					\path (axis cs:0,0) node [anchor=north west,yshift=-0.075cm,xshift=-0.075cm] {0};
					
				}
				
				\end{axis}
				
				\end{tikzpicture} 
				\caption{Optimal Subsidy}
		\end{subfigure}%
		\begin{subfigure}[]{.5\textwidth}
			\begin{tikzpicture}
	
				\begin{axis}[
					legend style={nodes={scale=0.8, transform shape}},
					axis y line=middle, 
					axis x line=middle,
					y axis line style={opacity=0},
					ytick={0},
					xtick={0,1},
				]
				\addplot [
					thick,
					name path=mte,
					domain=0:1, 
					samples=100, 
					color=black,
				]
				{cos(deg(x)*3.1415) + 0.25};
				\addlegendentry{MTE Curve};
				\path[name path=axis] (axis cs:0,0) -- (axis cs:1,0);
				\addplot [
					thick,
					color=blue,
					fill=blue, 
					fill opacity=0.5
				]
				fill between[
					of=mte and axis,
					soft clip={domain=0:0.58},
				];
				\addlegendentry{Positive Welfare};
				\addplot [
					thick,
					color=red,
					fill=red, 
					fill opacity=0.5
				]
				fill between[
					of=mte and axis,
					soft clip={domain=0.58:0.7},
				];
				\addlegendentry{Negative Welfare};
				
				\addplot[black, mark=diamond*] coordinates{(0.7,0)} node[above, color=black] {$g(\pi)$};
			
				after end axis/.code={
					\path (axis cs:0,0) node [anchor=north west,yshift=-0.075cm,xshift=-0.075cm] {0};
					
				}
				
				\end{axis}
				
			\end{tikzpicture}
			\caption{Non-optimal Subsidy}
		\end{subfigure}
		
	\end{center}
	\caption{Welfare of Subsidy Rules under Positive Selection}
	\label{fig:monotone-MTE}
	\caption*{\footnotesize Graph (a) shows the optimal subsidy when the MTE is monotonically decreasing. The necessary condition $\text{MTE}(g(\pi^*)) = 0$ is also sufficient because the MTE curve has a unique zero. Graph (b) shows that if we move $g(\pi)$ to the right, there will be a red region, and hence the welfare decreases.}
\end{figure}

\section{Welfare Properties of Subsidy Rules}\label{sec:welfare_property}

In most studies, instrumental variables are typically used to identify the treatment effects. In this section, we argue that the instrumental variable has a more fundamental influence on the policy design problem through its function of providing incentives for the treatment take-up. We show that assigning subsidies weakly dominates assigning treatments directly and can achieve the first-best welfare when the MTE is decreasing. 

To elaborate, we condition our analysis on $X = x$ throughout this section, implying that $g$ and $\pi$ are only functions of the instrumental variables $W$ and $Z$, and MTE is only a function of $u$. For simplicity, we assume that the action space $\mathcal{Z}^p$ is equal to the support of $Z$.

\subsection{Subsidies better than mandate}

To compare welfare, we introduce \emph{direct policies}, a new class of policies, that differ from the subsidy rules. The direct policies do not manipulate the subsidy $Z$. Instead, they directly manipulate the treatment take-up. 
Mathematically, a direct policy is a function $\tau : \text{Supp}(W,Z) \rightarrow \{0,1\}$. For an individual with characteristics $(w,z)$, if $\tau(w,z) = 1$, then the policy-maker makes the treatment mandatory. If $\tau(w,z) = 0$, the individual cannot select the treatment.\footnote{The result in this section can be generalized to allow for the randomization of direct policies. That is, the range of $\tau$ can be convexified to $[0,1]$. For simplicity, we do not consider this convexification in the propositions.} Denote the counterfactual outcome under the direct policy $\tau$ by 
\begin{align*}
	Y^\tau \equiv \tau(W,Z)Y_1 + (1-\tau(W,Z)) Y_0 .
\end{align*}

We study the following optimal \textit{identified} welfares under two policy settings:\footnote{In this section, we omit the cost part of the welfare as it is ambiguous to compare the costs of treatment and subsidy without a specific context.} 
\begin{align*}
	S^*_{{\text{sub}}} & \equiv \sup_{\pi: \text{Supp}(W,Z) \rightarrow \mathcal{Z}}  \mathbb{E} [Y^\pi \mathbf{1}\{U_D \in \text{Supp}(g(W,Z)) \}], \\
	S^*_{\text{dir}} & \equiv \sup_{\tau: \text{Supp}(W,Z) \rightarrow \{0,1\}} \mathbb{E} [ Y^\tau\mathbf{1}\{U_D \in \text{Supp}(g(W,Z)) \}].
\end{align*}
The subscript ``sub'' represents ``subsidy,'' and $S^*_{{\text{sub}}}$ is the identified welfare under subsidy rules. The subscript ``dir'' represents ``direct,'' and $S^*_{\text{dir}}$ is the optimal welfare under direct policies. We restrict the comparison of optimal welfare on the set of individuals whose $U_D$ lies in the region $\text{Supp}(g(W,Z))$, on which the treatment effect can be identified. The instrumental variables do not affect the treatment choice of individuals with $U_D$ outside the region $\text{Supp}(g(W,Z))$.\footnote{These individuals are referred to as always-takers or never-takers in the LATE literature.} We exclude these individuals in welfare comparisons because the data are inherently uninformative on the treatment effects and the counterfactual welfare under different policies for these individuals. The next two propositions provide a ranking between the two optimal welfares $S^*_{{\text{sub}}}$ and $S^*_{{\text{dir}}}$.

\begin{proposition} [Subsidies Better Than Direct Policies] \label{prop:value_instrument}
	Suppose that Assumptions \ref{assu:random_assignment}-\ref{assu:density_existence} hold. Then, $S^*_{{\text{sub}}} \geq S^*_{\text{dir}}$. 
	 The inequality holds strictly if the supremum in the definition of $S^*_{{\text{sub}}}$ is achieved through a unique policy $\pi^*$, such that $g(W,\pi^*(W))$ lies in the interior of $\text{Supp}(g(W,Z))$ with positive probability.
\end{proposition}

Proposition \ref{prop:value_instrument} states that the optimal welfare under subsidy rules is always preferred to that under direct policies. Subsidy-based policy weakly dominates the direct policy because the former affects the treatment status through changes in the treatment selection: all else being equal, a small subsidy incentivizes only individuals with low $U_D$ while larger subsidy incentivizes both low- and high-$U_D$ into the treatment status. Intuitively, the subsidy-based policy uses $U_D$ as a targeting variable although $U_D$ is unobservable. This capability to implicitly target with $U_D$ will enhance the welfare because $U_D$ correlates with $(U_1, U_0)$ and thus the potential outcomes.

Proposition \ref{prop:value_instrument} can be mathematically explained. Recall that from the discussion succeeding Proposition \ref{prop:welfare_representation}, the welfare of a subsidy rule is essentially the net area between the MTE and the cost function from zero to a nontrivial upper bound determined by the subsidy rule. As shown in the proof and in Theorem 1 of \cite{sasaki2020welfare}, we represent the welfare of a direct policy as an integral of MTE, but the integration region is the unit interval $[0,1]$.\footnote{The cost is not explicitly modeled in \cite{sasaki2020welfare}. Therefore, the welfare representation is the net area between the MTE and the horizontal axis.} Mathematically, direct policies do not have control over the integration region and therefore not as flexible as subsidy rules. We graphically demonstrate this argument at the end of this section.

We focus on the welfare implications of targeting the instruments $(W,Z)$ in direct policies. The Example \ref{eg:toy_example} in Section \ref{sec:welfare_opt_policy} shows that targeting $W$ is useful in designing subsidy rules. However, when considering direct policies, targeting $W$ or $Z$ does not improve the welfare because $W$ and $Z$ are (conditionally) independent with $(U_1,U_0,U_D)$ and are excluded from the outcome equation. After the treatment probability is assigned, the variation in the instruments is irrelevant to welfare. To formally state this result, we introduce a subclass of direct policies that have constant treatment probability. We call a direct policy $\tau$ a \emph{constant policy} if $\tau(w,z)$ does not vary with $(w,z)$. The optimal welfare for constant policies is denoted by 
\begin{align*}
	S^*_{\text{con}} & \equiv \sup_{\text{constant policy } \tau} \mathbb{E} [Y^\tau \mathbf{1}\{U_D \in \text{Supp}(g(W,Z)) \}].
\end{align*}

\begin{proposition} [Irrelevance of Instruments in Direct Policies] \label{prop:irrelevance_instrument}
	Suppose that Assumptions \ref{assu:random_assignment}-\ref{assu:density_existence} hold. Then, $S^*_{\text{dir}} = S^*_{\text{con}}$. 
\end{proposition}

\subsection{Subsidies achieve first-best welfare}

Although subsidy rules have the power to implicitly target on unobserved heterogeneity, targeting with subsidies is not optimal if the policy-maker could observe $U_D$ because the subsidy-based policy only targets individuals in a \emph{second-best} sense. The targeting is restricted to a specific form that individuals with low $U_D$ have to be in the treatment status whenever the high $U_D$ individuals are. However, we next show that in the specific case of positive selection, the optimal subsidy rule achieves the \emph{first-best} welfare. That is, the policy-maker cannot further improve the welfare from the optimal subsidy rule even if $U_D$ is observed.

To explain the definition of first-best, we should consider a situation where all the characteristics of an individual are observable, that is, the policy-maker has the power to assign the propensity of treatment based on $(W,Z,U_D)$. We want to keep in mind that targeting on $U_D$ is not feasible in practice. We are only considering such infeasible policies for welfare comparisons. Mathematically, an \emph{infeasible policy} is a function $\tilde{\tau}: \text{Supp}(W,Z) \times [0,1] \rightarrow \{0,1\}$. The policy-maker uses the information about $(W,Z,U_D)$ to mandate the treatment choice for each individual. Denote the counterfactual outcome under the direct policy $\tau$ by 
\begin{align*}
	Y^{\tilde{\tau}} \equiv \tilde{\tau}(W,Z,U_D)Y_1 + (1-\tilde{\tau}(W,Z,U_D))Y_0 .
\end{align*}
The first-best welfare is the optimal welfare with respect to the infeasible policies:
\begin{align*}
	S^*_{\text{fb}} & =  \sup_{\tilde{\tau}: \text{Supp}(W,Z) \times [0,1] \rightarrow \{0,1\}}   \mathbb{E} [Y^{\tilde{\tau}} \mathbf{1}\{U_D \in \text{Supp}(g(W,Z))],
\end{align*}
where the subscript ``fb'' represents ``first-best.'' The following proposition states that subsidy-based policies can achieve the first-best welfare.

\begin{proposition} [Subsidies can be First-best under Positive Selection] \label{prop:first_best}
	Suppose that Assumptions \ref{assu:random_assignment}-\ref{assu:density_existence} and \ref{assu:positive_selection} hold. Then, $S^*_{{\text{sub}}} = S^*_{\text{fb}}$. 
\end{proposition}

The underlying assumption behind Proposition \ref{prop:first_best} is that a subsidy-based policy can replicate the treatment assignment of the infeasible policy when the MTE is decreasing. Specifically, the infeasible first-best policy assigns all individuals with $\text{MTE}(u) \geq 0$ to the treatment group. Let $u_x^*$ be a solution to $\text{MTE}(u) = 0$. As individuals with higher MTE are always induced first in the case of positive selection, the subsidy-based policy can achieve the same counterfactual treatment choice if the individuals with $U_D = u_x^*$ are indifferent about the treatment choice under the policy, implying that individuals with $\text{MTE}(u) \geq 0$ are induced to the treatment group.



\subsection{Graphical Illustration}

To demonstrate the welfare comparisons using graphs, we impose the same simplifying assumptions as in Section \ref{subsec:graph}. We further assume that the support $\text{Supp}(g(W,Z)) = [0,1]$. Under these assumptions, there are two direct policies $\tau=0$ and $\tau=1$. Figure \ref{fig:welfare-direct} demonstrates the welfare of the two direct policies, which is either $0$ or the integral of the MTE on the entire unit interval $[0,1]$. Figure \ref{fig:welfare-characterization} shows that subsidy rules can freely select the integration region. Namely, the MTE is integrated over $[0,g(\pi)]$. Unless $\pi^*$ is a corner solution that belongs to $\{0,1\}$, we can choose a subsidy $\pi$ with $g(\pi) \in (0,1)$ that achieves a strictly higher welfare than the two direct policies.

Here, an infeasible policy $\tilde{\tau}$ is a function from $\text{Supp}(U_D) = [0,1]$ to $\{0,1\}$. To maximize the welfare, the policy-maker would want to mandate anyone with nonnegative MTE$(U_D)$ into treatment and exclude others from the treatment. The optimal infeasible policy $\mathbf{1}\{\text{MTE}(u) \geq 0\}$ achieves the first-best welfare. Figure \ref{fig:infeasible-welfare} shows that the optimal infeasible policy becomes a (feasible) subsidy rule when the MTE decreases.

\begin{figure}[!htbp]
	\begin{center}
		\begin{subfigure}[]{0.5\textwidth}
			\begin{tikzpicture}
	
				\begin{axis}[
					legend style={nodes={scale=0.8, transform shape}},
					axis y line=middle, 
					axis x line=middle,
					y axis line style={opacity=0},
					ytick={0},
					xtick={0,1},
				]
				\addplot [
					thick,
					name path=mte,
					domain=0:1, 
					samples=100, 
					color=black,
				]
				{0.399691 +   4.956392 * x  -74.935120*x^2 + 275.423206*x^3 -426.853331*x^4+  298.923460*x^5  -78.211517*x^6};

				after end axis/.code={
					\path (axis cs:0,0) node [anchor=north west,yshift=-0.075cm,xshift=-0.075cm] {0};
					
				}
				
				\end{axis}
				
				\end{tikzpicture}
				\caption{Welfare of direct policy $\tau = 0$}
		\end{subfigure}%
		\begin{subfigure}[]{0.5\textwidth}
			\begin{tikzpicture}
	
				\begin{axis}[
					legend style={nodes={scale=0.8, transform shape}},
					axis y line=middle, 
					axis x line=middle,
					y axis line style={opacity=0},
					ytick={0},
					xtick={0,1},
				]
				\addplot [
					thick,
					name path=mte,
					domain=0:1, 
					samples=100, 
					color=black,
				]
				{0.399691 +   4.956392 * x  -74.935120*x^2 + 275.423206*x^3 -426.853331*x^4+  298.923460*x^5  -78.211517*x^6};
				\addlegendentry{MTE Curve};
				\path[name path=axis] (axis cs:0,0) -- (axis cs:1,0);
				\addplot [
					thick,
					color=blue,
					fill=blue, 
					fill opacity=0.5
				]
				fill between[
					of=mte and axis,
					soft clip={domain=0:0.2},
				];
				\addlegendentry{Positive Welfare};
				\addplot [
					thick,
					color=red,
					fill=red, 
					fill opacity=0.5
				]
				fill between[
					of=mte and axis,
					soft clip={domain=0.2:0.5},
				];
				\addlegendentry{Negative Welfare};
				\addplot [
					thick,
					color=blue,
					fill=blue, 
					fill opacity=0.5
				]
				fill between[
					of=mte and axis,
					soft clip={domain=0.5:0.815},
				];
				\addplot [
					thick,
					color=red,
					fill=red, 
					fill opacity=0.5
				]
				fill between[
					of=mte and axis,
					soft clip={domain=0.815:0.99},
				];
			
				after end axis/.code={
					\path (axis cs:0,0) node [anchor=north west,yshift=-0.075cm,xshift=-0.075cm] {0};
					
				}
				
				\end{axis}
				
				\end{tikzpicture}
				\caption{Welfare of direct policy $\tau = 1$}
		\end{subfigure}

	\end{center}
	\caption{Welfare Characterization}
	\label{fig:welfare-direct}
	\caption*{\footnotesize Graph (a) shows the welfare of $\tau = 0$, which is equal to zero because we assume that $Y_0 = 0$. Graph (b) shows the welfare of $\tau = 1$, which is equal to the integral of MTE over $[0,1]$. A comparison of these two graphs with Figure \ref{fig:welfare-characterization} shows that direct policies are two extremum cases of subsidy rules. Therefore, the latter is superior in terms of welfare. This is the content of Proposition \ref{prop:value_instrument}.}
\end{figure}

\begin{figure}[!htbp]
	\begin{center}
	\begin{subfigure}[]{.5\textwidth}
		\begin{tikzpicture}
	
			\begin{axis}[
				legend style={nodes={scale=0.8, transform shape}},
				axis y line=middle, 
				axis x line=middle,
				y axis line style={opacity=0},
				ytick={0},
				xtick={0,1},
			]
			\addplot [
				thick,
				name path=mte,
				domain=0:1, 
				samples=100, 
				color=black,
			]
			{0.399691 +   4.956392 * x  -74.935120*x^2 + 275.423206*x^3 -426.853331*x^4+  298.923460*x^5  -78.211517*x^6};
			\path[name path=axis] (axis cs:0,0) -- (axis cs:1,0);
			\addplot [
				thick,
				color=blue,
				fill=blue, 
				fill opacity=0.5
			]
			fill between[
				of=mte and axis,
				soft clip={domain=0:0.2},
			];
			\addplot [
				thick,
				color=blue,
				fill=blue, 
				fill opacity=0.5
			]
			fill between[
				of=mte and axis,
				soft clip={domain=0.495:0.815},
			];
		
			after end axis/.code={
				\path (axis cs:0,0) node [anchor=north west,yshift=-0.075cm,xshift=-0.075cm] {0};
				
			}
			
			\end{axis}
			
		\end{tikzpicture}
		\caption{Non-monotonic MTE}
	\end{subfigure}%
	\begin{subfigure}[]{.5\textwidth}
		\begin{tikzpicture}
	
			\begin{axis}[
				legend style={nodes={scale=0.8, transform shape}},
				axis y line=middle, 
				axis x line=middle,
				y axis line style={opacity=0},
				ytick={0},
				xtick={0,1},
			]
			\addplot [
				thick,
				name path=mte,
				domain=0:1, 
				samples=100, 
				color=black,
			]
			{cos(deg(x)*3.1415) + 0.25};
			\addlegendentry{MTE Curve};
			\path[name path=axis] (axis cs:0,0) -- (axis cs:1,0);
			\addplot [
				thick,
				color=blue,
				fill=blue, 
				fill opacity=0.5
			]
			fill between[
				of=mte and axis,
				soft clip={domain=0:0.58},
			];
			\addlegendentry{First-best Welfare};
		
			after end axis/.code={
				\path (axis cs:0,0) node [anchor=north west,yshift=-0.075cm,xshift=-0.075cm] {0};
				
			}
			
			\end{axis}
			
			\end{tikzpicture}
			\caption{Decreasing MTE}
	\end{subfigure}
		
	\end{center}
	\caption{First-best Welfare under Infeasible Policies}
	\label{fig:infeasible-welfare}
	\caption*{\footnotesize The blue region indicates the first-best welfare achieved by the (infeasible) first-best policy $\tilde{\tau}(u) = \mathbf{1}\{\text{MTE}(u) \geq 0\}$. Graph (a) shows the first-best welfare under a general non-monotonic MTE curve. By explicitly targeting $U_D$, the policy-maker can avoid the red regions that appear in Figure \ref{fig:welfare-characterization}. Here, the first-best welfare is strictly greater than the welfare of any subsidy rules. In (b), the MTE is decreasing. The first-best welfare is equal to the welfare achieved by the optimal subsidy rule, as shown in Figure \ref{fig:monotone-MTE}. This is the content of proposition \ref{prop:first_best}. }
\end{figure}

\section{Identifying the Welfare Ranking}\label{sec:weflare_ranking}

In this section, we consider the identification of MTE and address the issue of identifying the optimal subsidy rule. We first discuss two scenarios in which the optimal policy can be point-identified. We then identify partial ranking among subsidy rules. 

By identification, we represent the objects of interest, such as welfare ranking and optimal policy, by the joint distribution of $(Y,D,X,W,Z)$. The analysis is conducted under the full knowledge of observable distributions, namely, the joint distribution of $(Y,D,X,W,Z)$. The welfare ranking $\succsim$ is an ordering on the policy space $\Pi$ such that $\pi\succsim\pi'$, if and only if $S(\pi) \geq S(\pi')$. The identified ranking is deemed partial if for some pairs of policies $(\pi,\pi')$ it is impossible to determine whether $S(\pi) \geq S(\pi')$ or $S(\pi) \leq S(\pi')$, given the observable distribution.

\subsection{Point identification of the optimal subsidy}


The welfare representation result indicates that both the welfare ranking and optimal policy are identified if the entire MTE curve and propensity score $g$ are identified. The following corollary presents the results.

\begin{corollary}
If for some $(x,w)\in\text{Supp}(X,W)$, $\text{MTE}(x,\cdot)$ is identified on $[0,1]$ and $g(x,w,\cdot)$ is identified on  $\mathcal{Z}^p$, then the optimal subsidy $\pi^*(x,w)$ is also identified.
\end{corollary}
\cite{heckman2005structural} showed that the MTE can be identified using \emph{local instrumental variables} (LIV). For completeness, we restate the MTE identification result . Define a function $m$ by
\begin{align*}
	m(x,u) \equiv \frac{d}{dp} \Big|_{p = u} \mathbb{E} [Y \mid X = x, g(X,W,Z) = p] \cdot \mathbf{1}_{\text{Supp}(g(x,W,Z))}(u).
\end{align*}
The function $m$ can be identified from the data, given that as the propensity score is identified on its support. The following lemma from \cite{heckman2005structural} states that $m$ is equal to the MTE on $\textit{Supp}(X,g(X,W,Z))$.

\begin{lemma}[Identification of MTE] Suppose that Assumptions \ref{assu:random_assignment}-\ref{assu:density_existence} hold. Further, assume that $g(X,W,Z)$ is a non-degenerate random variable conditional on $X$ and that $0<P(D = 1|X)<1$.
Then $\text{MTE}(x,u) = m(x,u), \text{ for all } (x,u) \in \text{Supp}(X,g(X,W,Z))$.
\end{lemma}

Lemma 1 shows that, to identify the entire MTE curve, the support of the propensity score must cover the unit interval for every $x\in\text{Supp}(X)$. Effectively, a large enough exogenous variation in the propensity score induced by $(W,Z)$ is needed. In literature, this is termed as the large support assumption. Identification of the propensity score $g$ on $\text{Supp}(X,W,Z)$ is straightforward because it is simply the observed probability of take-up, given the covariates and instruments. When $Z^p\not\subset Z$, $g$ on $\text{Supp}(X,W,Z)$ cannot be identified by imposing parametric restriction on the propensity score $g$ such as the probit model.

Point identification of the MTE curve is unnecessary for identifying the optimal policy in two scenarios: First, if the policies under consideration assign subsidies only from the support of $Z$, that is, when $\mathcal{Z}^p \subset \mathcal{Z}$, then the optimal policy is identified. As shown in the next proposition, it is possible to identify the optimal policy without identifying MTE. Define $\mathscr{P}^{id} = \{\pi\in\Pi: \pi(x,w) \in \text{Supp}(Z \mid X=x, W=w) \;\text{for all}\;(x,w)\in \text{Supp}(X,W)\}$ as the set of identiable subsidy rules. 

\begin{proposition}[Identification on the Support] \label{prop:identification_observed}  Suppose that Assumptions \ref{assu:random_assignment}-\ref{assu:density_existence} hold. 
	Then, for any $\pi \in\mathscr{P}^{id}$, 
		\begin{align} \label{eqn:id_observed_z}
		\mathbb{E}[Y^{\pi}\mid X,W] = \mathbb{E}[Y\mid X,W,Z=\pi(X,W)],
		\end{align}
		and
		\begin{align}
		  \mathbb{E}[C^\pi\mid X,W] &= \pi(X,W) \cdot \mathbb{E}[D\mid X,W,Z=\pi(X,W)]],
	\end{align}
where $\mathbb{E}[Y|X,W,Z=\pi(X,W)]$ and $\mathbb{E}[D\mid X,W,Z=\pi(X,W)]$ are identified as $\pi(x,w) \in \text{Supp}(Z \mid X=x, W=w)$. Thus, the welfare ranking on $\mathscr{P}^{id}$ is identified.
\end{proposition}

Proposition \ref{prop:identification_observed} states that the counterfactual welfare can be identified by empirical welfare provided that the subsidy under consideration is observed in the data. 

A second scenario in which no point identification of the MTE curve is needed is when individuals positively select into the treatment status. Under the assumption of positive selections, individuals with higher returns are always induced first by the subsidies. Therefore, an amount of subsidy is optimal if the marginal effect of subsidy is zero.
\begin{proposition}[Identification under Positive Selection]
\label{prop:id_positive_seleciton} Suppose the assumptions stated in Proposition \ref{prop:positive_seleciton} hold. If there exists $z^*\in \mathcal{Z}^p$ such that $g(x,w,z^*)$ and $\text{MTE}(x,g(x,w,z^*))$ are identified and that $\Lambda(x,w,z^*) = 0$, then $\pi^*(x,w) = z^*$.
\end{proposition}
Proposition \ref{prop:id_positive_seleciton} states that, if the selection is positive, the optimal amount subsidy can be identified provided that point at which the marginal effect is zero is known. Therefore, no instruments are needed to have large support if it contains the point having a zero marginal effect. Even if the requirement is not met, imposing positive selection still has identification power, as shown in the next subsection.

\subsection{Partial ranking of subsidy rules}

While the results in the previous subsection yielded point identification, the requirements can be restrictive. Therefore, the method for obtaining a partial ranking of subsidy rules by imposing shape restrictions is discussed. 

Let $\mathcal{M}^o$ and $\mathcal{G}^o$ be sets of functions that represent, respectively, the functional parameter space of MTE and propensity under possible shape restrictions (e.g., parametric model, monotonicity, and boundedness). The true MTE is assumed to be an element of $\mathcal{M}^o$ and the true propensity is assumed to be an element of $\mathcal{G}^o$. Moreover, the true MTE must coincide with the identifiable function $m$ on the identified region. Therefore, the identified set $\mathcal{M}$ of MTEs under shape restrictions is
\begin{align*}
	\mathcal{M} = & \big\{ \bar{m}(x,u) \in  \mathcal{M}^o:  \bar{m}(x,u) = m(x,u) \text{ for all } (x,u) \in \text{Supp}(X,g(X,W,Z)) \big\} .
\end{align*}
The identified set $\mathcal{G}$ of propensities is
\begin{align*}
	\mathcal{G} = & \left\{ g \in \mathcal{G}^o : g(x,w,z) = \mathbb{E}[D\mid X=x,W=w,Z=z], (x,w,z) \in \text{Supp}((X,Z,W)) \right\}.
\end{align*}
Let $\langle \cdot, \cdot \rangle$ be the inner product with respect to the measure that underlies the random vector $(X, U_D)$. Define the dual cone and polar cone of $\mathcal{M}$, respectively, as
\begin{align} \label{eqn:dual_cone}
	\mathcal{M}^* & = \{ \ell : \langle l,\bar{m} \rangle \geq 0, \bar{m} \in \mathcal{M} \} \text{, and }
	\mathcal{M}^\times = -\mathcal{M}^*.
\end{align}

For any subsidy-based policy $\pi$ and propensity $g$, we use $F_{g,\pi}(x,u)$ to denote the conditional cumulative distribution function (CDF) of the propensity score $g(X,W,\pi(X,W))$ given $X = x$, that is,  $F_{g,\pi}(x,u) \equiv \mathbb{P}(g(X,W,\pi(X,W)) \leq u\mid X = x)$.

\begin{proposition} [Identification of Partial Ranking] \label{prop:id_ranking}
Suppose that Assumptions \ref{assu:random_assignment}-\ref{assu:density_existence} hold. For simplicity, assume that the cost $c = 0$. Let $(\pi,\pi')$ be a pair of subsidy rules. If $\{F_{g,\pi'} - F_{g,\pi}: g \in \mathcal{G}\} \subset \mathcal{M}^*$, then $\pi \succsim \pi'$. If $\{F_{g,\pi'} - F_{g,\pi}: g \in \mathcal{G}\} \subset \mathcal{M}^\times$, then $\pi' \succsim \pi$. 
\end{proposition}

Proposition \ref{prop:id_ranking} identifies a partial ranking among the subsidy rules. This result is the subsidy-rules analog of Proposition 1 in \cite{kasy2016partial}, who studied the identification of a partial ranking among direct policies. The partial identification is achieved by considering the geometry in the Hilbert space containing data-consistent MTE curves. For a pair of subsidy rules, if the difference between the induced CDFs of the propensity score is orthogonal to the set of plausible MTEs, then the data are uninformative about the welfare ranking between the policies. From the decision-theoretic perspective, the identified welfare ranking constitutes an incomplete ordering that admits an expected utility representation \citep{dubra2004expected}. 

Our next proposition states that when there is positive selection, the partial identification result for the policy ranking can be easily interpret.
 
\begin{proposition}[Direction of Welfare Improvement]\label{prop:partial_monotone} Suppose the assumptions in Proposition \ref{prop:positive_seleciton} hold, and let $\pi^*$ be an optimal policy. If $\Lambda(x,g(x,w,z)) \geq 0$  (resp. $\leq0$) for some $z \in \mathcal{Z}^p$, then $\pi^*(x,w) \geq z$ (resp. $\leq z$).

\end{proposition} 
The same intuition behind Proposition \ref{prop:positive_seleciton} applies here: Under positive selection, the marginal benefit of subsidy decreases. Recall that the function $\Lambda(x,w,z)$ only depends on the MTE and the propensity score, and it represents the marginal return of subsidy. Therefore, if the MTE and the propensity score are identified at a specific point $(x,w,z)$, then the optimal subsidy can be bounded from below if the marginal return at that point is positive and can be bounded from above when it is negative.

\section{Empirical Application}\label{sec:empirical}
Our method is illustrated by applying it to the experimental data from the Jordan New Opportunities for Women (Jordan NOW) pilot study. In the experiment, vouchers for wage subsidies are randomly assigned to female college students in their last year of education. These vouchers can be presented to firms while seeking jobs, and if a student with a voucher is employed, the employer can redeem the voucher for up to six months for an amount equal to the minimum wage. The premise of the program is that wage subsidies can help students land their first jobs in which they can acquire experience and skills that can help in their long-term careers. We refer our readers to \cite{groh2016wage} for more details about the background and their experiment design.

In their study, \cite{groh2016wage} found that, although wage subsidies substantially increase the employment rate immediately after graduation, their effects on the long-term labor market participation is limited. Specifically, wage subsidies increase the employment rate by about 38 percentage points during the subsidized period. However, the effect vanishes rapidly after the subsidy expires. Seventeen months after the subsidies expire, the effect on employment is less than two percentage points and not statistically significant.  \cite{groh2016wage} concluded that providing wage subsidies is not an effective measure to promote women's long-term labor market participation, at least in the context of Jordan.

In our empirical exercise, we investigate the effectiveness of wage subsidies can be improved through targeting. The welfare function considered is the 30-month earnings ($Y$) after the subsidy period minus the cost of wage subsidy.\footnote{We proxy the 30-month earnings by the monthly earnings reported at the last round, which occurred two years after the voucher had expired.} In the model, $Y$ is the realization of one of the potential outcomes $Y_1$ and $Y_0$, depending on whether the student successfully found a job after graduation during the subsidized period ($D = 1$) or not ($D = 0$). In our policy exercise, we target based on the student's college major. Specifically, we target on whether the student majors in medical assistance ($X$), which includes nursing and pharmacy specializations. We search for the optimal subsidy within the space $\mathcal{Z}^p = [0,900]$, where $z^p = 0$ refers to no subsidy and  $z^p = 900$ is the maximal subsidy an individual could receive in the experiment. We use voucher cost in Example \ref{ex:voucher_cost}.

The optimal amount of subsidy critically depends on two factors: (1) how effectively can wage subsidy encourage and help students to find their first job and (2) the treatment effect of having a first job after graduation on the long-term labor market outcome. These effects are traditionally quantified by imposing joint normality on the error terms $(U_1,U_0,U_D)$ and their independence to $(X,Z)$. Then, the outcome and choice equations are estimated together with either the method of maximum likelihood or the two-step method proposed by \cite{heckman1976common}. For this illustration, we will proceed accordingly. Although these parametric assumptions could be restrictive, they yield accurate estimates when the sample size is modest. We provide a more flexible method that does not impose normality in Appendix \ref{app:semiparametric}. 

Formally, we estimate the following selection model:
\begin{align*}
Y_1 &= X'\beta_1 + U_1, \\
Y_0 &= X'\beta_0 + U_0, \\
D &= \mathbf{1}\{X'\beta_D + Z\gamma - \tilde{U}_D \geq 0\},\\
(X,&Z) \perp (U_1,U_0,\tilde{U_D}), (U_1,U_0,\tilde{U_D})  \sim \mathcal{N}(0, \Sigma),
\end{align*} 
where 
\[
 \Sigma = \begin{pmatrix}
\sigma_1^2 & \rho_{01}\sigma_0\sigma_1 & \rho_1\sigma_1\\
\rho_{01}\sigma_0\sigma_1 & \sigma_0^2 & \rho_0\sigma_0  \\
\rho_1\sigma_1 & \rho_0\sigma_0  & 1
\end{pmatrix}. 
\]

\begin{table}[!htbp]    
	\captionsetup{skip=0pt}
\caption{Estimates of the Heckman selection model}
\begin{center}
\begin{threeparttable}
\begin{tabular}{c *{3}{d{3.3}} }
    \toprule
    \toprule  
    \textbf{} &\multicolumn{1}{c}{ \textbf{Choice equation}} & \multicolumn{1}{c}{\textbf{Outcome equation ($Y_0$)}} & \multicolumn{1}{c}{\textbf{Outcome equation $(Y_1)$}} \\
\midrule
      Medical major    &      $0.2965$       &    $1743.9040$     &  $2677.209$\\
    			       &     $(0.0894)$      &   $(351.3193)$      &  $(420.9792)$  \\
      Subsidy   	       &      $0.0017$        &                          &  \\ 
     		               &      $(0.0001)$      &                          &  \\
      Intercept            &    $-0.9359$        & $607.5856$          & $660.1336$ \\ 
		               &     $(0.0741)$       & $(192.6821)$       & $(437.1357)$\\ \midrule
     $\rho_d$             &                         & $-0.0889$ 	   & $0.3802$ \\
     $\sigma_d$	       & 		 	&   $2596.7705$    & $3399.0894$ \\
\bottomrule
\end{tabular}
\begin{tablenotes}
\item \footnotesize The sample size is 1347. Standard errors are obtained by bootstraps. Wage subsidies significantly increase the probability of finding a job.
\end{tablenotes}
\end{threeparttable}
\end{center}
\label{tbl:heckman} 
\end{table}

\begin{figure}[!htbp]    
	\centering
    \includegraphics[width=0.6\textwidth]{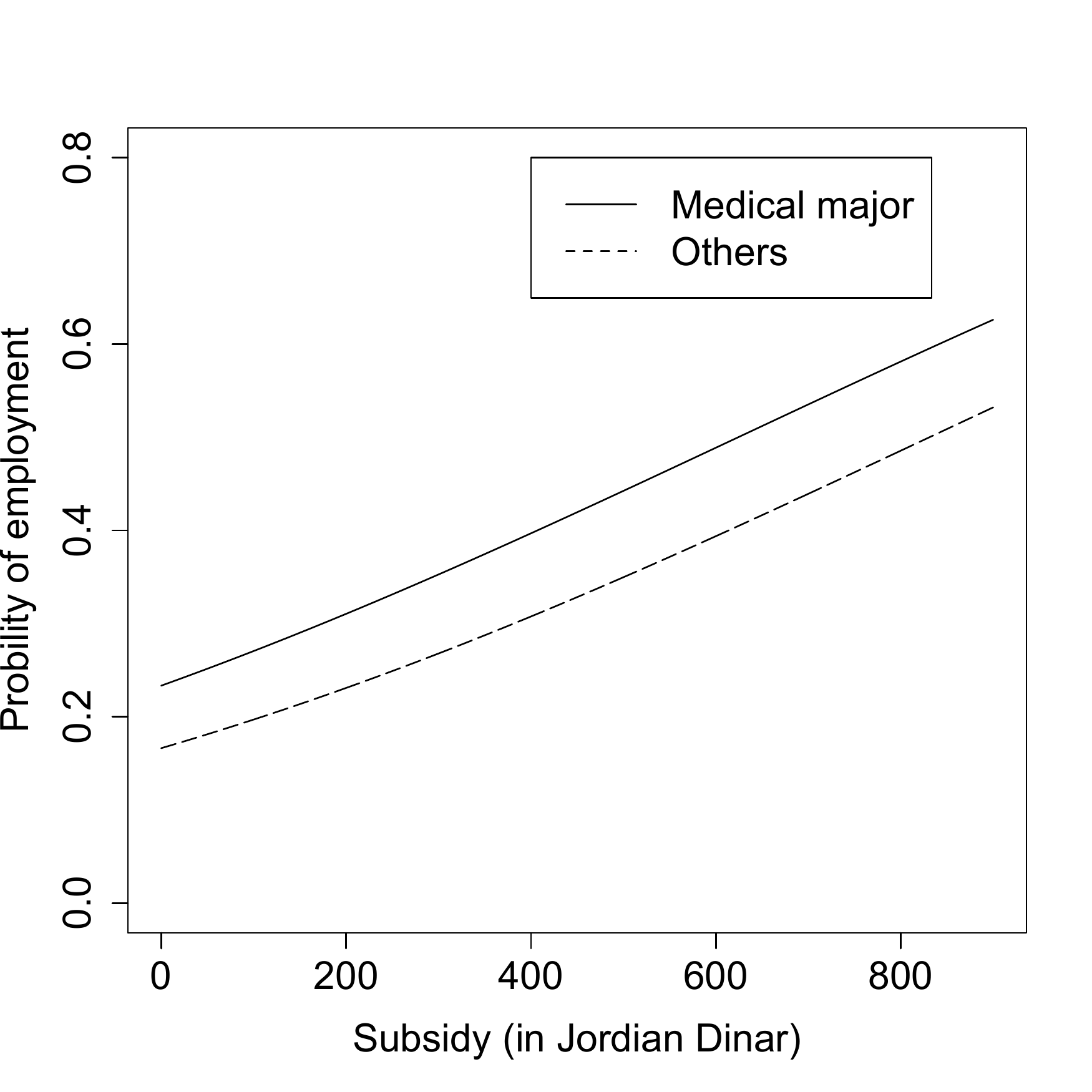}
    \caption{The fitted probabilities from the choice equation}
	\caption*{\footnotesize The propensity score is increasing in the amount of subsidy. The employment rate almost triples after receiving subsidies. Medical students are more likely to be employed at any given level of subsidy. }
     \label{fig:propensity}
\end{figure}

Table \ref{tbl:heckman} presents the estimates from the Heckman two-step method. Similar to the results in \cite{groh2016wage}, we find that wage subsidies significantly increase the chance of finding a job after graduation. In Figure \ref{fig:propensity}, we plot the probabilities as functions of subsidies. Notably, at the maximal amount of subsidy, the employment rate almost triples compared with the case with no subsidy. Moreover, the employment rate is higher among students with medical majors at any given level of subsidy. From the policy-maker's perspective, this implies that helping medical students to land their first job is less expensive. 

\begin{figure}[!htbp]    
	\centering
    \includegraphics[width=0.6\textwidth]{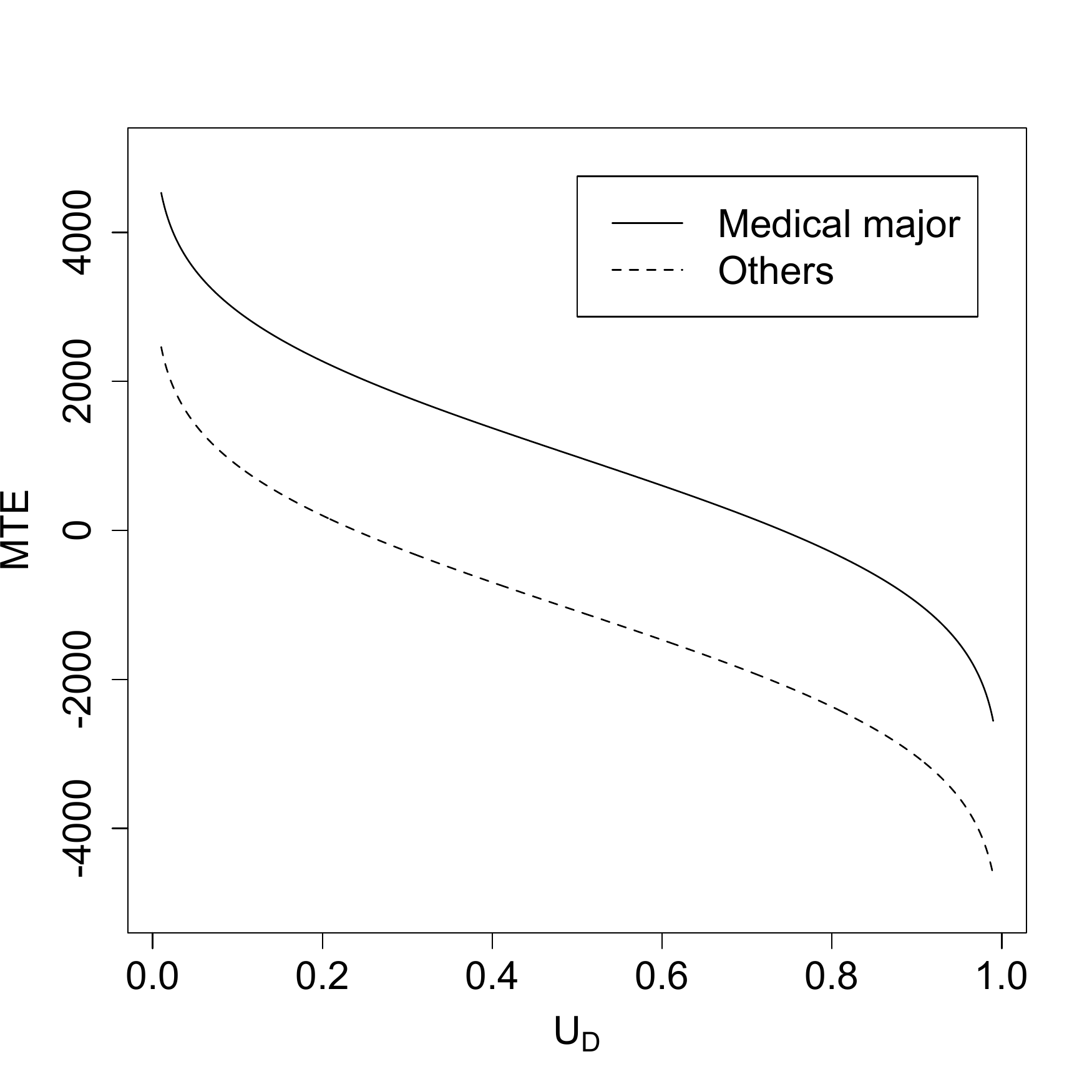}
    \caption{MTE estimated from a normal selection model}
	\caption*{\footnotesize The MTEs are decreasing, indicating that there is positive selection. Individuals who are more likely to find a job after graduation tend to benefit more from it for their long-term career prospects.}
    \label{fig:mte}
\end{figure}

In the normal selection model, \cite{heckman2003simple} showed that the MTE is 
\[
\text{MTE}(x,u) = x'(\beta_1 - \beta_0) - (\rho_1\sigma_1 - \rho_0\sigma_0)\Phi^{-1}(u).
\]
We plot the estimated MTE curves in Figure \ref{fig:mte}. Remarkably, first, the MTE curves are downward-sloping, suggesting that individuals positively select into the treatment. Specifically, individuals who are more likely to find a job after graduation tend to benefit more from it for their long-term career prospects. Second, the marginal effect can be negative among individuals with high $U_D$, meaning that wage subsidies can be potentially hard for individuals with low willingness to work.\footnote{A possible explanation for the negative effect is the stigmatization toward voucher users \citep{burtless1985targeted}.} From the policy-maker's perspective, it is cheaper to help medical students land their first job. 

\begin{figure}[!htbp]
\begin{subfigure}{.5\textwidth}
  \centering
  \includegraphics[width=\linewidth]{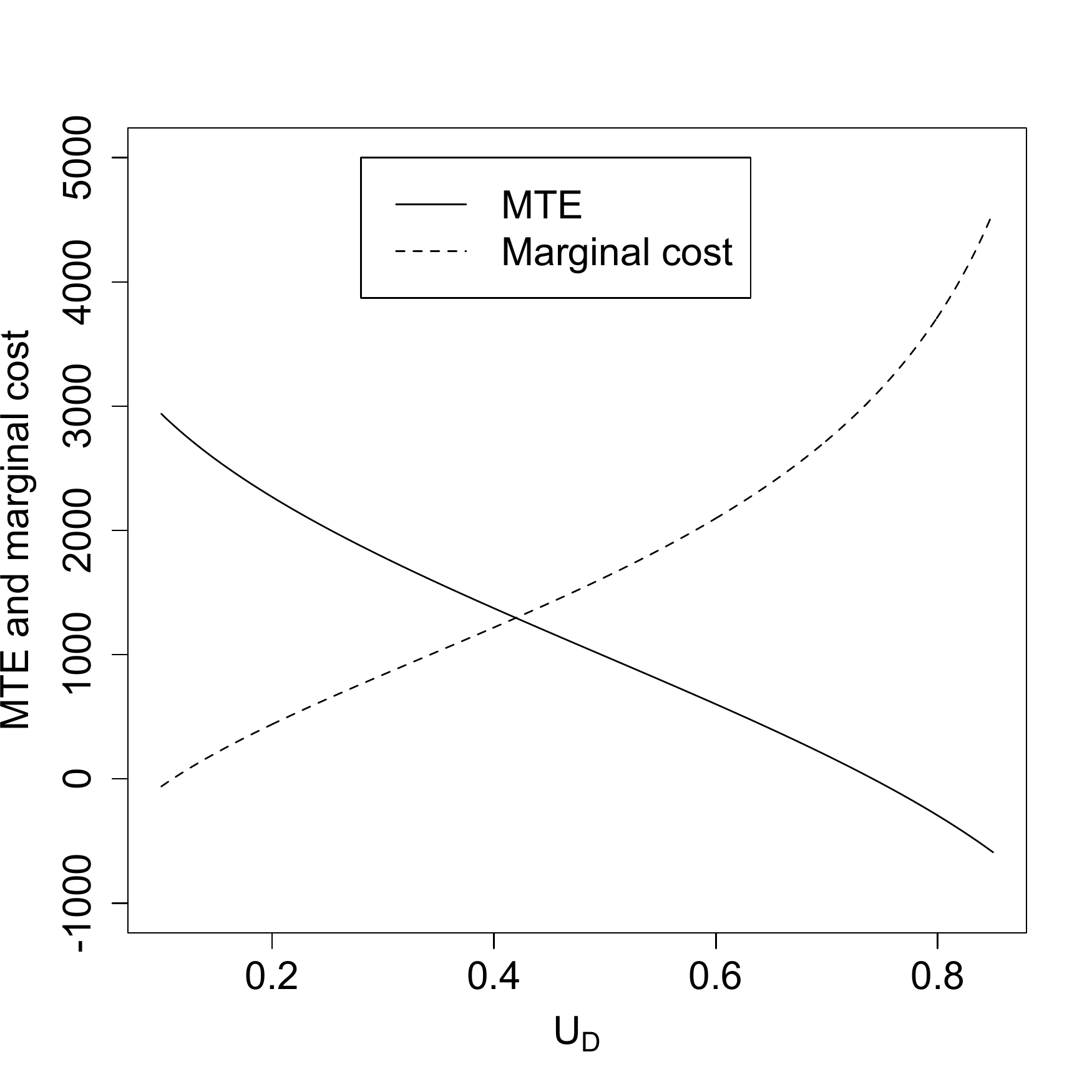}
  \caption{Medical majors}
\end{subfigure}%
\begin{subfigure}{.5\textwidth}
  \centering
  \includegraphics[width=\linewidth]{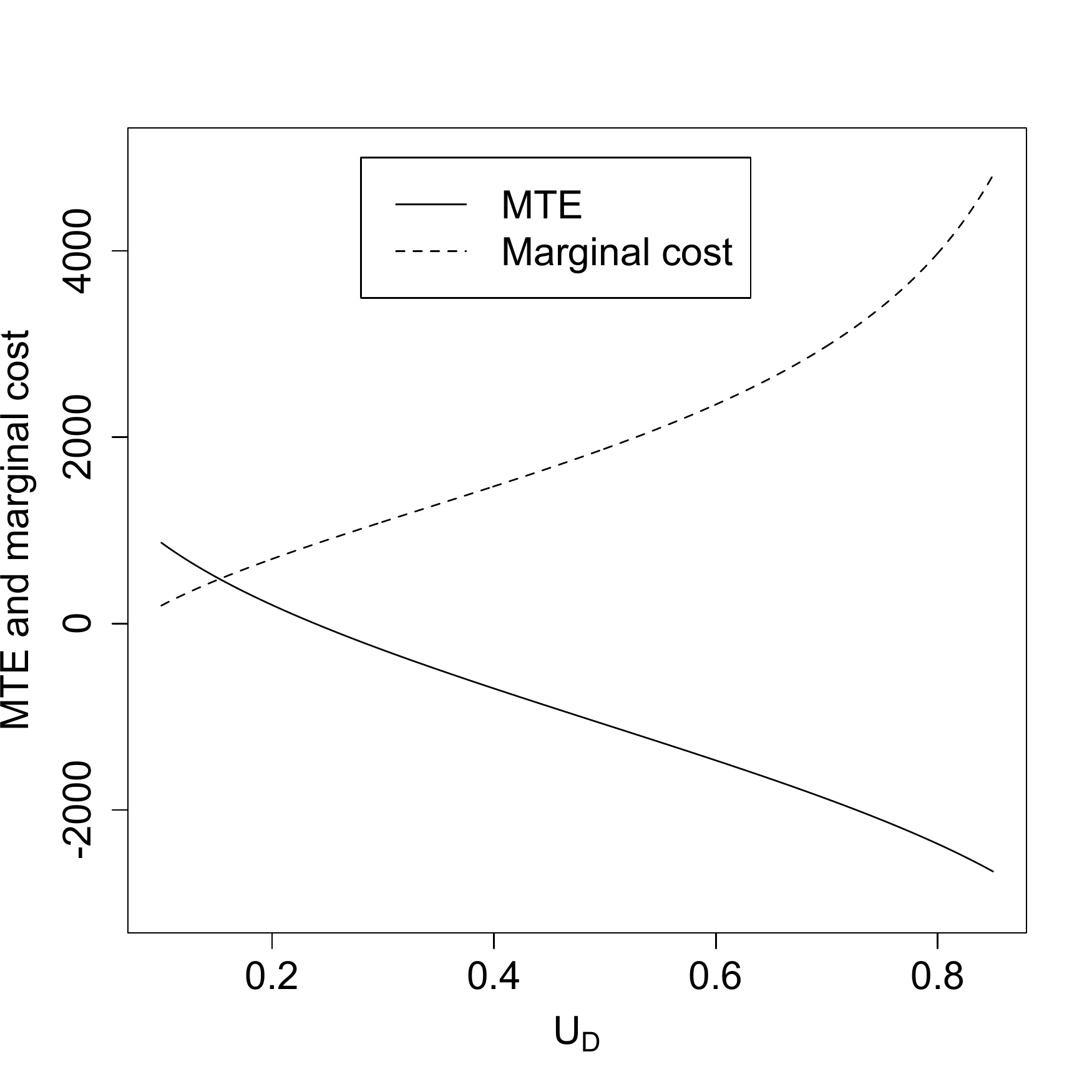}
  \caption{Other majors}
\end{subfigure}
\caption{MTE curves and marginal costs for different majors}
\caption*{We find the optimal subsidies by locating the intersection of the MTE and the marginal cost. The optimal subsidies substantially differ for the two groups of students. Medical students tend to benefit more from wage subsidies.}
\label{fig:policy}
\end{figure}

Combining the previous results, we plot the MTE curves and marginal cost of subsidies in Figure \ref{fig:policy}. Based on Proposition \ref{prop:optimality_condition}, we examine the optimal take-up rate at the intersection of the MTE and marginal cost curves. The optimal subsidies substantially differ for the two groups of students. As discussed earlier, as students with medical majors tend to benefit more from subsidies, the optimal take-up rate is higher than that for students from other majors. In fact, by substituting the optimal take-up rate into the inverse of propensity scores, we find that the optimal subsidy for medical students is about JOD $375$, which is approximately one-third of the amount provided in the experiment. However, regarding students from other majors, the estimated optimal subsidy is negative (JOD $-75$). As negative subsidies are excluded in this context, we have a corner solution, which implies that the policy-maker should not provide subsidies for this group of students. 

Our results suggest that the subsidy set in the experiments is much higher than their optimal level (in terms of the welfare function we defined). The insignificant long-term effect found in \cite{groh2016wage} may be a result of excessive amounts of subsidies that may draw individuals with lower returns. The welfare outcome may improve by lowering the amount of subsidy. Moreover, its targeting efficiency can be further enhanced to exploit the heterogeneity in treatment effects and take-ups.

\section{Conclusion} \label{sec:conclusion}
In this study, we examine the problem of allocating subsidies based on individual characteristics. We adopt the MTE framework to analyze the characterization, identification, and welfare properties of subsidy rules. Our results show that subsidy rules generally outperform policies that directly mandate the treatment. In our empirical example, we estimate the optimal wage subsidy using a parametric MTE model. More flexible methods that do not impose parametric assumptions are provided in the appendix. Their theoretical properties are of interest for future studies.

\appendix

\numberwithin{equation}{section}
\numberwithin{proposition}{section}

\section{Proofs of the Results in the Main Text} \label{sec:proofs}

\begin{proof} [Proof of Proposition \ref{prop:welfare_representation}]
	Given that $Y = (1-D)Y_0 + DY_1$, we have
	\begin{align*}
		\mathbb{E} [ Y^\pi ] & = \mathbb{E} [ D^\pi(Y_1 -Y_0) ] + \mathbb{E} \left[ Y_0 \right].
	\end{align*}
	By the law of iterated expectations 
	\begin{align*}
		\mathbb{E} \left[ D^\pi (Y_1 -Y_0) \right] & = \mathbb{E}\left[  \mathbb{E} \left[ \mathbf{1}\{g(X,W,\pi(X,W))  \geq U_D\}(Y_1 -Y_0) \mid U_D,X,W \right]\right]  \\
		& = \mathbb{E}\left[  \mathbb{E} \left[ \mathbf{1}\{g(X,W,\pi(X,W))  \geq U_D\} \mid U_D,X,W \right] \mathbb{E}\left[(Y_1 -Y_0) \mid U_D,X \right]\right] \\
		& = \mathbb{E} \left[ \mathbf{1}\{g(X,W,\pi(X,W)) \geq U_D\} \text{MTE}(X,U_D) \right] \\
		& = \mathbb{E} \left[ \mathbb{E} \left[ \mathbf{1}\{g(X,W,\pi(X,W)) \geq U_D\} \text{MTE}(X,U_D) \mid X,W \right] \right] \\
		& = \mathbb{E} \left[ \int_0^{g(X,W,\pi(X,W))} \text{MTE}(X,u) du \right]. 
	\end{align*}
	The second line follows from the fact that the indicator $\mathbf{1}\{g(X,W,\pi(X,W))  \geq U_D\}$ is $\sigma(X,W,U_D)$-measurable and hence can be taken outside the conditional expectation operator. The third line follows from the definition of MTE in (\ref{eqn:def-MTE}). The last line follows from the normalization that $U_D | X,W$ follows the uniform distribution on $[0,1]$.
	For the cost function, we have
	\begin{align*}
	\mathbb{E}\left[C^\pi\right] 
	& = \mathbb{E}\left[c(X,W,\pi(X,W),D^\pi)\right] \\
	& = \mathbb{E}\left[\mathbb{E}\left[c(X,W,\pi(X,W),D^\pi)\mid X,W\right]  \right],\\
	&= \mathbb{E}\left[c(X,W,\pi(X,W),1)\cdot g(X,W,\pi(X,W))\right] \\
	& \hspace{8pt}+ \mathbb{E}\left[c(X,W,\pi(X,W),0)\cdot (1 - g(X,W,\pi(X,W)))\right] 
\end{align*}
where the last equality follows from $\mathbb{E}\left[D\mid X,W,Z\right] = g(X,W,Z)$ and that $D$ is binary.
\end{proof}

\begin{proof} [Proof of Proposition \ref{prop:optimality_condition}]
Considering that $g_{x,w}$ is a one-to-one mapping from $\mathcal{Z}$ to the unit interval $[0,1]$, maximizing over $z$ can be equivalently solved by maximizing over $u$ with change of variables. Formally, following the expressions derived in the proof of Proposition \ref{prop:welfare_representation}, we have
\begin{align*}
\mathbb{E}[Y^\pi \mid X = x,W = w] 
&= \mathbb{E}[Y_0\mid X = x,W = w] +  \int_0^{g(x,w,\pi(x,w))} \text{MTE}(x,u') du'\\
&= \mathbb{E}[Y_0\mid X = x,W = w] + \int_0^{u_{x,w}} \text{MTE}(x,u') du',
\end{align*}
where $u_{x,w} = g_{x,w}(\pi(x,w)) = g(x,w,\pi(x,w))$. Similarly, we can write
	\begin{align*}
	\mathbb{E}\left[C^\pi\mid X=x,W=w\right] 
	&= c(x,w,g_{x,w}^{-1}(u_{x,w}),1)\cdot u_{x,w} + c(x,w,g_{x,w}^{-1}(u_{x,w}),0)\cdot (1 - u_{x,w}).
\end{align*}
As $g_{x,w}$ is assumed to be a one-to-one mapping, for each $X = x, W = w$, we can find the welfare-maximizing subsidy $\pi(x,w) \in[z_l,z_u]$ by maximizing the welfare over $u_{x,w}\in{I_{x,w}}=\{ g(x,w,z): z \in [z_l,z_u] \}$.

Then, the first-order condition for the optimization problem can be found by differentiating $\mathbb{E}[Y^\pi - C^\pi]$ with respect to $\pi^*(x,w) \in[z_l,z_u]$ and substituting $u_{x,w} = g_{x,w}(\pi(x,w))$ as follows:
	\begin{align*} 
	\begin{split}
		&\text{MTE}(x,u^*_{x,w}) \cdot \left.\frac{d}{dz}g_{x,w}(z) \right|_{z = g_{x,w}^{-1}(u^*_{x,w})}\\
		&= u^*_{x,w}\cdot\left.\frac{d}{dz} c(x,w,z,1)\right|_{z=g_{x,w}^{-1}(u^*_{x,w})} + c(x,w,g_{x,w}^{-1}(u_{x,w}),1) \cdot \left.\frac{d}{dz}g_{x,w}(z) \right|_{z = g_{x,w}^{-1}(u^*_{x,w})} \\
		+ & (1 - u^*_{x,w})\cdot\left.\frac{d}{dz} c(x,w,z,0)\right|_{z=g_{x,w}^{-1}(u^*_{x,w})} - c(x,w,g_{x,w}^{-1}(u^*_{x,w}),0) \cdot \left.\frac{d}{dz}g_{x,w}(z) \right|_{z = g_{x,w}^{-1}(u^*_{x,w})}.
	\end{split}
	\end{align*} 
After calculating the corresponding derivatives for $c(x,w,z,d) = z\cdot d$, we obtain the results of Proposition \ref{prop:optimality_condition}.
\end{proof}

\begin{proof} [Proof of Proposition \ref{prop:positive_seleciton}]
Following the proof of Proposition \ref{prop:optimality_condition}, the first-order derivative and the optimal take-up rate problem defined in Equation (\ref{eqn:optim_prob}) is
\[
\text{MTE}(x,u) - g_{x,w}^{-1}(u) - u  \cdot \frac{d}{du}g_{x,w}^{-1}(u).
\] 
Direct calculations yield that the second-order derivative is
\[
\frac{\partial}{\partial u}MTE(x,u) - 2 \cdot [\left.\frac{d}{du} g_{x,w}(u)]^{-1}\right|_{u=g^{-1}(u)} - u\cdot \frac{d^2}{du^2} g_{x,w}^{-1}(u).
\]
By applying the formula for derivatives of inverse functions, it is obvious that the second-order derivative is guaranteed to be non-positive because $\frac{\partial}{\partial u}MTE(x,u) \leq 0$ (Assumption \ref{assu:positive_selection}), $\frac{d}{dz}g_{x,w}(z) \geq 0$ (Assumption \ref{assu:invertibility}), and $\frac{d^2}{dz^2}g_{x,w}(z) \leq 0$.
The results follow from that the objective function of the optimal take-up rate problem is concave.
\end{proof}

\begin{proof} [Proof of Proposition \ref{prop:negative_seleciton}]
From Proposition \ref{prop:welfare_representation}, we know that the objective function of the optimal take-up rate problem (\ref{eqn:optim_prob}) is given by
\[
f(u_{x,w}) = \int_0^{u_{x,w}} \text{MTE}(x,u')du'
\]
and that its first-order derivative is MTE. Therefore, by Assumption \ref{assu:negative_selection}, the objective function is convex as the MTE is increasing, and its solution is either the left endpoint $\underline{u}_{x,w} = g(x,w,z_l)$ or the right endpoint $\bar{u}_{x,w} = g(x,w,z_u)$ of the feasible region $I_{x,w}$. To complete the proof, observe that  
\begin{align*}
	f(\bar{u}_{x,w}) - f(\underline{u}_{x,w}) = \int^{\bar{u}_{x,w}}_{\underline{u}_{x,w}} \text{MTE}(x,u') du' = \int^{g(x,w,z_u)}_{g(x,w,z_l)} \text{MTE}(x,u') du'.
\end{align*}

\end{proof}

\begin{proof} [Proof of Proposition \ref{prop:value_instrument} and \ref{prop:irrelevance_instrument}]
	Without the loss of generality, we assume that $Y_0 = 0$ in the proof of propositions \ref{prop:value_instrument}, \ref{prop:irrelevance_instrument}, and \ref{prop:first_best}.
	Following the proof of Proposition \ref{prop:welfare_representation}, we have
	\begin{align} \label{eqn:S*-sub}
		S^*_{{\text{sub}}} & = \sup_\pi \{ \mathbb{E} \left[ \mathbf{1}\{g(W,\pi(W)) \geq U_D, U_D \in \text{Supp}(g(W,Z))\} \text{MTE}(U_D)  \right] \} \\
		& = \sup_\pi  \mathbb{E} \left[  \int_{B_\pi(W)} \text{MTE}(u) du  \right],
	\end{align}
	where $B_\pi(w) = [0,g(w,\pi(w))] \cap \text{Supp}(g(W,Z))$.
	In contrast, by the independence between $W$ and $(U_1,U_0,U_D)$, we have
	\begin{align*}
		S^*_{\text{dir}} & = \sup_{\tau \text{ direct policy}} \left\{ \mathbb{E} \left[ \tau(W,Z)(Y_1 - Y_0)\mathbf{1}\{U_D \in \text{Supp}(g(Z,W))\}  \right]  \right\} \\
		& = \sup_{\tau \text{ direct policy}} \mathbb{E}[\tau(W,Z)] \mathbb{E} \left[  \int_{\text{Supp}(g(W,Z))} \text{MTE}(u) du  \right] , \\
		& = \sup_{\tau \in \{0,1\}} \tau \mathbb{E} \left[  \int_{\text{Supp}(g(W,Z))} \text{MTE}(u) du  \right] ,
	\end{align*}
	The last equality follows from the fact that $\tau$ only affects the welfare through the expectation $\mathbb{E}[\tau(W,Z)]$. As we can choose $\pi$ such that $\text{Supp}(g(W,Z)) \subset B_\pi$, we have $S^*_{{\text{sub}}} \geq S^*_{\text{dir}}$. However, the optimal welfare achieved by constant policies is equal to 
	\begin{align*}
		S^*_{\text{con}} = \sup_{\tau \in \{0,1\}} \tau \mathbb{E} \left[  \int_{\text{Supp}(g(W,Z))} \text{MTE}(u) du  \right] ,
	\end{align*}
	which is equal to $S^*_{\text{dir}}$.
	If the supremum in the definition of $S^*_{{\text{sub}}}$ is achieved by a unique policy $\pi^*$ such that $g(W,\pi^*(W))$ lies in the interior of $\text{Supp}(g(W,Z))$ with positive probability, then $S^*_{{\text{sub}}} > S^*_{\text{dir}}$.
\end{proof}

\begin{proof} [Proof of Proposition \ref{prop:first_best}]
	The welfare of an infeasible policy	$\tilde{\tau}$ is
	\begin{align*}
	& \mathbb{E} \left[ Y^{\tilde{\tau}} \mathbf{1}\{U_D \in \text{Supp}(g(W,Z)) \} \right] \\
	= & \mathbb{E}\left[\mathbb{E} \left[ (Y_1-Y_0) \widetilde{\tau}(W,Z,U_D)\mathbf{1}\{U_D \in \text{Supp}(g(W,Z)) \}  \mid W,Z, U_D\right]\right] \\
	= & \mathbb{E}\left[\tilde{\tau}(W,Z,U_D)\mathbf{1}\{U_D \in \text{Supp}(g(W,Z)) \} \mathbb{E} \left[ (Y_1-Y_0)   \mid X, W, U_D\right]\right] \\
	= & \mathbb{E}\left[\tilde{\tau}(W,Z,U_D)\mathbf{1}\{U_D \in \text{Supp}(g(W,Z)) \} \cdot \text{MTE}(U_D)\right] , 
	\end{align*}
where the first equality is an application of the law of iterated expectation and the third equality holds as $(Z,W) \perp (U_1,U_0, U_D) \mid X$. From the above equation, it is clear that the optimal infeasible policy assigns the treatment status whenever $\text{MTE}(u) \geq 0$. Therefore, $\tilde{\tau}^*(u) =\mathbf{1}\{\text{MTE}(u)\geq 0\}$. Because MTE is decreasing, we have $\{u\mid\text{MTE}(u)\geq 0\} = [0,u^*]$ for some $u^*$. Then, the first-best (identified) welfare, $S^*_{\text{fb}}$, is equal to the integral of the MTE over the region $[0,u^*] \cap \text{Supp}(g(W,Z))$. Subsequently, the result follows by observing the expression of $S^*_{\text{sub}}$ in (\ref{eqn:S*-sub}).

\end{proof}


\begin{proof} [Proof of Proposition \ref{prop:identification_observed}]
	Considering that $Y^\pi = D^\pi Y_1 + (1-D^\pi)Y_0$, we have
		\begin{align*} 
		\mathbb{E} [ D^\pi Y_1 \mid X,W] 
		& = \mathbb{E} [ D^\pi Y\mid X,W] 	\\	
		& = \mathbb{E} [ \mathbf{1}\{g(X,W,\pi(X,W))  \geq U_D\} Y\mid X,W] \\
		& = \mathbb{E} [ \mathbf{1}\{g(X,W,Z)  \geq U_D\} Y\mid X,W,Z = \pi(X,W)] \\
		& = \mathbb{E} [ DY\mid X,W,Z = \pi(X,W)].
		\end{align*}
		
		Similarly, we show that $ \mathbb{E} [ (1-D^\pi) Y_1 \mid X,W] =  \mathbb{E} [ (1-D)Y\mid X,W,Z = \pi(X,W)]$. Combining the two steps, we have
		\begin{align*}
		\mathbb{E}[Y^{\pi}\mid X,W]
		& = \mathbb{E}[D^\pi Y_1 + (1-D^\pi)Y_0\mid X,W] \\
		& = \mathbb{E}[DY+ (1-D)Y\mid X,W,Z = \pi(X,W)] \\
		& = \mathbb{E}[Y \mid X,W,Z = \pi(X,W)].
		\end{align*}
The second part of the proposition follows from Proposition \ref{prop:welfare_representation} and that 
	\begin{align*}
		 \mathbb{E}[D^\pi \mid X,W] = \mathbb{E}[D\mid X,W,Z = \pi(X,W)]. 
	\end{align*}
		Therefore, we have
	\begin{align*}
	\begin{split}
		  \mathbb{E}[C(X,W,\pi(X,W),D) \mid X,W] &= 
		  C(X,W,\pi(X,W),0)(1-\mathbb{E}[D\mid X,W,Z=\pi(X,W)])\\
		  & + C(X,W,\pi(X,W),1)\mathbb{E}[D\mid X,W,Z=\pi(X,W)].
	\end{split}
	\end{align*}
	
\end{proof}

\begin{proof}[Proof of Proposition \ref{prop:id_positive_seleciton}]
The result follows directly from Proposition \ref{prop:positive_seleciton}.
\end{proof}

\begin{proof} [Proof of Proposition \ref{prop:id_ranking}]
		
		In the proof of proposition \ref{prop:welfare_representation}, we have shown that
		\begin{align*}
			\mathbb{E} \left[ D^\pi (Y_1 -Y_0) \right] = \mathbb{E} \left[ \mathbf{1}\{g(X,W,\pi(X,W)) \geq U_D\} \text{MTE}(X,U_D) \right].
		\end{align*}
		By the law of iterated expectations and the assumption that $U_D \perp (X,W)$, we have
		\begin{align*}
			\mathbb{E} \left[ D^\pi (Y_1 -Y_0) \right] & = \mathbb{E} \left[\mathbb{E} \left[ \mathbf{1}\{g(X,W,\pi(X,W)) \geq U_D\} \mid X, U_D \right] \right] \\
			& = \mathbb{E} \left[ \text{MTE}(X,U_D) \mathbb{E}\left[\mathbf{1}\{g(X,W,\pi(X,W)) \geq U_D\} \mid X,U_D\right]\right]\\
			& = \mathbb{E} \left[ \text{MTE}(X,U_D) (1-F_{g,\pi}(X,U_D))\right]\\
			& = \langle 1-F_{g,\pi}, \text{MTE} \rangle.
		\end{align*}
The difference in welfare between the two policies $\pi$ and $\pi'$ is
\begin{align*}
S(\pi) - S(\pi') = \langle F_{g,\pi'}-F_{g,\pi}, \text{MTE} \rangle.
\end{align*}
		Then the result follows from the definitions of $\mathcal{M}^*$ and $\mathcal{M}^\times$ in Equation (\ref{eqn:dual_cone}).
\end{proof}

\begin{proof}[Proof of Proposition \ref{prop:partial_monotone}]
Fix $X=x$. Considering that the first-order derivative of the objective function is
\[
f'(u) = \frac{d}{du}\int_0^{u} \text{MTE}(x,u') du' = \text{MTE}(x,u),
\] 
MTE$(x,u_0) > 0$ implies that the objective is increasing at $u_0$. Combined with the fact that $f$ is concave when MTE$(x,\cdot)$ is decreasing, we know that 
\begin{align*}
	\argmax_{u\in[0,1]} \text{ } \int_0^{u} \text{MTE}(x,u') du' \geq u_0.
\end{align*}
The proof for the case of MTE$(x,u_0) < 0$ is similar.
\end{proof}

\section{Semiparametric and nonparametric methods of policy learning}
\subsection{Semiparametric method}\label{app:semiparametric}

Although the local IV approach for the estimation of MTE is fully nonparametric, it comes at the cost of slower convergence. In applications with smaller sample size, the nonparametric approach might be inappropriate. In this section, we describe a semiparametric approach, which is less data-demanding than the nonparametric method but more flexible than parametric methods.

First, to facilitate the estimation of MTE, it is common in the MTE literature \citep{carneiro2009estimating,carneiro2011estimating,brinch2017beyond,zhou2019marginal} to specify the potential outcomes in Equation (\ref{eqn:outcome_equation}) as a linear representation
\begin{align*}
	Y_1 = X \beta_1 + U_1 \text{, and } Y_0 = X \beta_0 + U_0,
\end{align*}
where $\beta_1$ and $\beta_0$ are the unknown parameters. Another often used useful simplification is to assume that $(X,W,Z)$ are jointly independent of $(U_1,U_0,U_D)$. Considering the earlier specification, we have
\begin{align*}
	\mathbb{E} \left[ Y \mid X = x, g(X,W,Z) = p \right] & = x\left( p\beta_1 + (1-p)\beta_0 \right) + \mathbb{E} \left[ (U_1 - U_0) \mathbf{1}\{U_D \leq p\} \right] \\
	& =  x\beta_0 + xp\left( \beta_1 -\beta_0 \right) + \lambda(p),
\end{align*}
where $\lambda(p) = \int_0^p \mathbb{E} \left[ U_1 - U_0 \mid U_D = u \right] du $ has an unknown form. Taking partial derivative with respect to $p$, we get a partial linear specification of the MTE:
\begin{align*} 
	\text{MTE}(x,u) = x(\beta_1 - \beta_0) + \lambda'(u).
\end{align*}
We consider the semiparametric estimators $\hat{\beta}_1$ and $\hat{\beta}_0$ defined in \cite{carneiro2009estimating}, which is based on the semiparametric estimation procedure developed by \cite{robinson1988root}. Under mild regularity conditions, they are $\sqrt{n}$-consistent.

The semiparametric approach relies on the parametrization of the selection correction term
\begin{align*}
	\lambda(p) = \lambda(p\mid\theta)
\end{align*}
and of the propensity score
\begin{align*}
    g(X,W,Z) = g(X,W,Z\mid\gamma),
\end{align*}
where $\theta$ and $\gamma$ are the parameters. The common choices of $\lambda(p)$ include finite-order polynomials \citep{brinch2017beyond} and the inverse of normal cumulative distribution function \citep{carneiro2011estimating}. Regarding the propensity score $g(X,W,Z\mid\gamma)$, popular choices include the logit and probit models.

The aforementioned parametrization implies that 
\begin{align*}
	\mathbb{E} \left[ Y \mid X = x, g(X,W,Z) = p \right] 
	& =  x\beta_0 + xp\left( \beta_1 - \beta_0 \right) + \lambda(p\mid\theta),
\end{align*}
which can be estimated as follows:

Step 1:	estimate $\gamma$ according to the specified binary choice model. For example, run probit regression if the probit model is assumed

Step 2: plug in the estimates $\hat{\gamma}$ to obtain the generated regressor $\hat{P} = g(X,W,Z\mid\hat{\gamma})$ 

Step 3: estimate $(\beta_1,\beta_0,\theta)$ by solving the nonlinear least-squares problem
\begin{align*}
    \min_{\beta_0,\beta_1,\theta} \sum_{i=1}^N [Y_i - X\beta_0 - X\hat{P}\beta_1 - \lambda(\hat{P}\mid\theta) ]^2
\end{align*}
If desired, additional constraints can be incorporated in the regression to guarantee that the MTE curve $\text{MTE}(x,u) = x(\beta_1 - \beta_0) + \lambda'(u)$ satisfies a certain shape restriction, such as monotonicity.

Step 4: for each $X=x$ and given the estimates $(\hat{\beta}_0,\hat{\beta}_1,\hat{\gamma})$, we can estimate the optimal participation rate $\hat{u}^*_{x,w}$ by solving the following equation:
\begin{align*} 
	\text{MTE}(x,u) = x(\beta_1 - \beta_0) + \lambda'(u) = 0.
\end{align*}
The optimal incentive assignment $z^*_{x,w}$ is then estimated by inverting the propensity score 
\begin{align*} 
	\hat{z}^*_{x,w} = g^{-1}_{x,w}(u^*),
\end{align*}
where $g_{x,w}(z) = g(x,w,z\mid\hat{\gamma})$.

Note that Steps 1-3 add up to a two-stage regression, which can be developed as a generalized method of moments estimator that is asymptotically normal under regularity conditions. Considering that $\hat{z}^*_{x,w}$ is a function of the estimates $(\hat{\beta}_0,\hat{\beta}_1,\hat{\gamma})$, we can derive its asymptotic properties based on the asymptotic distribution of $(\hat{\beta}_0,\hat{\beta}_1,\hat{\gamma})$ and the delta method. 

\subsection{Policy learning under monotonicity}
In this subsection, we briefly describe how monotonicity of MTE can aid the estimation for optimal policy and outline an estimation procedure. Specifically, the approach has the advantage of being fully automated as no tuning parameter is required, although the MTE was non-parametrically unspecified. The approach is simple, particularly when the propensity score is concave.\footnote{A third approach is to assume the monotone treatment response ($Y_1 \geq Y_0$) and apply the monotone regression of $Y$ on $Z$.} Notably, the method also serves as a new approach for estimating the MTE. 

When MTE is monotone, say, decreasing, the conditional mean $\mathbb{E}[Y|P]$ is concave in $P$. If the propensity score $P$ was known, concavity suggests that we can employ a concave regression of $Y$ on $P$. Considering that concave regression is uniformly consistent under regularity condition \citep{seijo2011nonparametric}, the optimal participation rate can also be consistently estimated from maximizing the fitted regression. Formally, let $Q_n$ be a nonparametric least-square estimator that solves
\begin{align*}
\min_{Q_n \text{ concave}} \sum_{i = 1}^{n} (Y_i - Q_n(P_i)).
\end{align*} 
Considering that $Q_n$ is only uniquely defined on the observed values of $P_i$, we linearly interpolate between points of $P_i$. Additionally, let $\partial Q_n (p)$ denote the subdifferentail of $Q_n$. \cite{seijo2011nonparametric} shows that, under mild conditions, $Q_n$ is uniformly consistent of $\mathbb{E}[Y|P]$. Furthermore, if $\mathbb{E}[Y|P]$ is differentiable, then the estimated subdifferential is also uniformly consistent with the derivative. As differentiability is typically assumed for the identification of MTE, we may construct a uniformly consistent estimator of the MTE curve from the convex regression.

As $Q_n$ is uniformly consistent, we can estimate the optimal participation rate by finding $p_n = arg\max_p Q_n(p)$.\footnote{Alternatively, we can obtain $p_n$ by solving for the root of the estimated MTE curve.} By inference, \citep{ghosal2017univariate} derived the asymptotic distribution of $p_n$, given the homoscedasticity and smoothness of $\mathbb{E}[Y|P]$.\footnote{In our context, their smoothness assumption requires that $MTE$ is not flat at its root and is continuously differentiable.} However, it is not trivial to construct a confidence interval based on the asymptotic theory because the limiting distribution is not pivotal and depends on the second-order derivative of $\mathbb{E}[Y|P]$.

So far, we have treated the propensity score as known, which is not the case in most applications. The existing theory on convex regression should be extended to apply the aformentioned approach. However, if both $E[Y|P]$ and $p(Z) = \mathbb{E}[D|Z]$ are concave, then $\mathbb{E}[Y|p(Z)]$ is also concave in $Z$. Therefore, we can estimate $\mathbb{E}[Y|Z]$ using concave regression and solve for the optimal policy by  maximizing the regression function $\mathbb{E}[Y|Z]$. The aforementioned theory of convex regression is readily applicable in this case. Note that this method is similar to an intent-to-treat approach with \citep{kitagawa2018should}, but the shape restriction that resulted from the MTE framework allows us to deal with non-discrete assignments. Specifically, the MTE framework provides a convenient yet a valid method in economics theory to impose structures over the instruments $Z$.

\section{Primitive Conditions for Monotone MTE}\label{sec:primitive_monotone}

This section provides a set of primitive conditions for the MTE to be monotonic in Example \ref{eg:general-roy}.

\begin{proposition} [Primitive Conditions for Monotone MTE] \label{prop:primitive_monotone}
	Let the treatment be determined by 
	\begin{align*}
		D = \mathbf{1}\{\phi(X,W,Z,\Delta,V) \geq 0\},
	\end{align*}
	where $(Z,W) \perp (\Delta, V)$ and $\Delta = Y_1 - Y_0$. If the function $\phi$ satisfies the following two conditions
	\begin{enumerate} [label = (\roman*)]
		\item $\phi$ is increasing (resp. decreasing) in $\Delta$, and
		\item for any $(x,w,z),(x',w',z')$, $\phi(x,w,z,\delta_0,v_0) > \phi(x',w',z',\delta_0,v_0)$ for some $\delta_0,v_0 \implies \phi(x,w,z,\delta,v) > \phi(x',w',z',\delta,v)$ for all $\delta,v$.
	\end{enumerate}
	Then we can construct (1) a random variable $U_D$, such that $(W,Z) \perp U_D \mid X$ and $U_D \mid X \sim \text{Unif}[0,1]$ and (2) a function $g$, such that the treatment selection is represented by
	\begin{align*}
		D = \mathbf{1}\{g(X,W,Z) \geq U_D\}.
	\end{align*}
	Under this construction, the MTE curve defined by MTE$(x,u) \equiv \mathbb{E} \left[ \Delta \mid X=x, U_D = u \right]$ is decreasing (resp. increasing) in $u$.
\end{proposition}

\begin{proof} [Proof of Proposition \ref{prop:primitive_monotone}]
	Consider $\phi$ to be increasing in $\Delta$. Based on \cite{vytlacil2006note}, we have $\phi(x,w,z,\delta,v) = \tilde{\phi}_2(\tilde{\phi}_1(x,w,z),\delta,v)$, where the functions $\tilde{\phi}_1$ and $\tilde{\phi}_2$ are constructed as follows:
	Select any $\delta_0,v_0$, then define $\tilde{\phi}_1(x,w,z) = \phi(x,w,z,\delta_0,v_0)$. Define a correspondence $\tilde{\phi}_2$ as
	\begin{align*}
		\tilde{\phi}_2(\cdot,\delta,v) = \{ \phi(x,w,z,\delta,v): \tilde{\phi}_1(x,w,z) = \cdot \}.
	\end{align*}
	It can be shown that $\tilde{\phi}_2$ is a single-valued function, strictly increasing in the first argument, and increasing in the second argument. Define $\tilde{\phi}_3(\delta,v)$ by $\tilde{\phi}_2(\tilde{\phi}_3(\delta,v),\delta,v) = 0$. $\tilde{\phi}_3$ is well-defined as $\tilde{\phi}_2$ is strictly increasing in its first argument. Then we have
	\begin{align*}
		D & = \mathbf{1}\{\phi(X,W,Z,\Delta,V) \geq 0\} \\
		& = \mathbf{1}\{\tilde{\phi}_2(\tilde{\phi}_1(X,W,Z),\Delta,V) \geq 0\} \\
		& = \mathbf{1}\{\tilde{\phi}_1(X,W,Z) \geq \tilde{\phi}_3(\Delta,V)\}.
	\end{align*}
	Note that $\tilde{\phi}_3$ is decreasing in $\delta$, because for any $v$, $\tilde{\phi}_2$ increases with $\delta$; so, by definition, $\tilde{\phi}_3$ must decrease when $\delta$ increases. Define $F_{3\mid X}$ as the conditional CDF of $\tilde{\phi}_3(\Delta,V) \mid X$. Then let $g(X,W,Z) = F_{3 \mid X}(\tilde{\phi}_1(X,W,Z))$, and $U_D = F_{3 \mid x}(\tilde{\phi}_3(\Delta,V))$. Thus, it holds that $U_D \perp (W,Z) \mid X$, $U_D \mid X \sim \text{Unif}[0,1]$, and $D = \mathbf{1}\{g(X,W,Z) \geq U_D\}$. The MTE 
	\begin{align*}
		\text{MTE}(x,u) & = \mathbb{E} \left[ \Delta \mid X=x, U_D = u \right] \\
		& = \mathbb{E} \left[ \Delta \mid X=x, F_{3 \mid x}(\tilde{\phi}_3(\Delta,V)) = u \right]
	\end{align*}
	is increasing in $u$ as both the function $F_{3 \mid X}$ and $\tilde{\phi}_3(\cdot,v)$ are increasing.
	
	\end{proof}

The function $\phi(X,W,Z,\Delta,V) $ represents the utility achieved. The individual will select into the treatment if and only if the utility is positive. Condition (i) means that the utility is monotonic in the individual treatment effect. In the study on return to schooling, $\Delta$ represents the change in earnings after receiving a certain level of education, in which case $\phi$ is increasing in $\Delta$. Condition (ii) means that the rank of the instrument based on individual's utility is invariant to the values of the individual treatment effect $\Delta$ and the unobserved heterogeneity $V$.

As a side note, under a rank-invariance condition in \cite{vytlacil2006note}, we can always find a transformation such that the treatment choice $g$ is increasing in the transformed value of the instrument. This would be helpful when the monotonicity of $g$ facilitates identification.

\bibliographystyle{chicago}
\bibliography{references.bib}

\end{document}